\documentclass[dvipsnames, 11pt]{article}
\usepackage{jheppub}

\usepackage[utf8]{inputenc}

\usepackage{amsmath}
\usepackage{amssymb}
\usepackage[all]{xy}
\usepackage{tikz}
\usetikzlibrary{calc}

\usepackage[toc,page]{appendix}
\usepackage{enumerate}
\usepackage{tensor}
\usepackage{booktabs}
\usepackage{bbm}
\usepackage{youngtab}
\usepackage{slashed}
\usepackage{multirow}
\usepackage{makecell}
\usepackage{float}
\usepackage{comment}
\usepackage{verbatim}
\usepackage{xcolor}
\usepackage[framemethod=tikz]{mdframed}
\usepackage{soul}
\usepackage{mathtools}
\usepackage[inline]{enumitem}
\usepackage{alphabeta}
\usepackage{accents}
\usepackage{graphicx}
\usepackage{pgfplots}
\usepgfplotslibrary{fillbetween}
\usetikzlibrary{decorations.markings}

\usepackage{amsthm}

\newtheorem{lemma}{Lemma}[section]
\newtheorem{corollary}{Corollary}[lemma]

\newcommand{\ab}{|}

\newcommand{\de}{\mathrm{d}}

\newcommand{\yr}{y_{\mathrm{r}}}
\newcommand{\yi}{y_{\mathrm{i}}}
\newcommand{\yrd}{\dot{y}_{\mathrm{r}}}
\newcommand{\yid}{\dot{y}_{\mathrm{i}}}

\newcommand{\m}{m}
\newcommand{\f}{f}

\newcommand{\In}{\mathrm{b}}
\newcommand{\fin}{0}
\newcommand{\M}{\mathrm{m}}
\newcommand{\rad}{\mathrm{r}}

\newcommand{\rec}{\mathrm{rec}}

\newcommand{\e}{\mathrm{e}}
\newcommand{\I}{\mathrm{i}}

\newcommand{\p}{\mathrm{P}}

\allowdisplaybreaks

\usepackage{hyperref}
\hypersetup{
    colorlinks=true,
    linkcolor=orange!65!purple,
    citecolor=orange!65!purple,
    filecolor=orange!65!purple,   
    urlcolor=orange!65!purple,
}

\title{Bounding axion dark energy}

\author[a]{Gary Shiu,}
\author[b]{Flavio Tonioni,}
\author[c]{Hung V. Tran}

\affiliation[a]{Department of Physics, University of Wisconsin-Madison, 1150 University Avenue, Madison, WI 53706, USA}
\affiliation[b]{Dipartimento di Fisica e Astronomia, Università degli Studi di Padova, and INFN-Padova, via F. Marzolo 8, 35131 Padua, Italy}
\affiliation[c]{Department of Mathematics, University of Wisconsin-Madison, 480
Lincoln Drive, Madison, WI 53706, USA}

\emailAdd{shiu@physics.wisc.edu}
\emailAdd{flavio.tonioni@unipd.it}
\emailAdd{hung@math.wisc.edu}

\abstract{
We study cosmological solutions of (pseudo)scalar theories with periodic potentials, in the presence of arbitrary cosmological fluids -- including a cosmological constant of either sign.
Independently of the initial misalignment angle and field velocity, we derive an analytic bound that the axion mass parameter and decay constant fulfill as the universe decreases its acceleration rate, finding a natural application in models of thawing quintessence.
As a first application, we illustrate the analytic handle our bound provides in bounding axion dark energy, after observational inputs from DESI and various supernovae data sets are taken into account.
As a second application, we argue that our analytic bounds in combination with proposed quantum gravity constraints on axions exclude vast regions of parameter space.
The combined constraints push the axion masses to be much larger than the Hubble scale, in tension with basic models of axion quintessence.
}

\begin{document}

\maketitle

\newpage
\pagenumbering{arabic}

\section{Introduction}

Since the discovery of the accelerated expansion
of our universe \cite{SupernovaCosmologyProject:1998vns, SupernovaSearchTeam:1998fmf}, fundamental 
physics has been confronted with the arduous task of providing a microscopic explanation of the dark energy (DE) that accounts for this phenomenon. 
Recently, the Dark Energy Spectroscopic Instrument (DESI) has raised the complexity to an even higher level by showing a tantalizing hint of a time-evolving dark energy \cite{DESI:2024mwx, DESI:2025zpo, DESI:2025zgx}.
Although the possibility of quintessence has been investigated for a long time \cite{Ratra:1987rm, Wetterich:1987fm, Caldwell:1997ii}, the question of finding a phenomenological model capable of reproducing 
observations is now more pressing than ever.

On the theoretical side, there has been a rising consensus that the parameters of effective field theories (EFTs) coupled to gravity are not freely chosen but are subject to constraints that originate in the ultraviolet (UV) completion of the theory (see refs.~\cite{Palti:2019pca, Agmon:2022thq, Montero:2024qml} for reviews).
Interestingly, one of the most consequential, albeit more conjectural, implications of these quantum gravity constraints is the absence of shallow positive potentials \cite{Obied:2018sgi, Ooguri:2018wrx, Bedroya:2019snp}.
This insight has led to intense activity towards quintessence model-building even prior to the DESI data release (DR) \cite{Danielsson:2018ztv, Agrawal:2018own, Heisenberg:2018yae, Cicoli:2018kdo, Raveri:2018ddi}.
Regardless of whether this conjectural property of scalar potentials in quantum gravity ultimately holds --  and how broadly it applies beyond the asymptotic regimes of field space -- a robust and quite generic prediction of string theory is the ubiquitous presence of axions (or axion-like particles) \cite{Svrcek:2006yi} (see ref.~\cite{Marchesano:2024gul} for a review of the role of axions in realistic string constructions).
The ubiquity of axions is coined the \emph{axiverse} \cite{Arvanitaki:2009fg}.
For this reason, axions have been a prime candidate to realize dynamical DE in string theory.
For reviews on string-theoretic attempts to realize DE, see e.g. refs.~\cite{Cicoli:2023opf, Andriot:2026lac}.

Axions have emerged as a prime candidate for quintessence from a string-theoretic perspective not only because of their ubiquity, but also because they naturally overcome several hurdles in quintessence model building.
The properties of axions that make them suitable for realizing dark energy
include the following.
\begin{enumerate*}[label=(\alph*)]
    \item The axion shift symmetries protect their potentials
    from radiative corrections.
    \item The axion potential is generated only non-perturbatively, making it plausible for their dynamics to emerge at exponentially small energy scales, while other fields can be stabilized perturbatively at hierarchically larger masses.
    \item The axion derivative couplings with other fields greatly relax the experimental lower bounds on their masses.
\end{enumerate*}
However, other problems remain.
Among others, to provide a sustained epoch of cosmic acceleration, the axion decay constant would have to be super-Planckian, much larger than is considered possible in consistent theories of quantum gravity \cite{Rudelius:2014wla, Rudelius:2015xta, Brown:2015iha, Brown:2015lia,Heidenreich:2015wga}.
As a result, the initial conditions of the axion would have to be extremely fine-tuned for cosmic acceleration to happen in the current epoch.

Motivated by the phenomenological context described above, in this paper we study the cosmological evolution of (pseudo)scalars with periodic potentials.
We are also able to account for potentials that have global minima of negative energy, and not just a vanishing one.
A DE model that eventually settles into an anti-de Sitter (AdS) vacuum would face less stringent theoretical constraints from its UV completion \cite{Danielsson:2018ztv}.
In ref.~\cite{Luu:2025fgw}, this scenario was also argued to provide a better fit with observational data.

The dynamics of axions coupled to Friedmann–Lema\^{i}tre–Robertson–Walker (FLRW) spacetime and other cosmological fluids is in general quite complicated.
Except for certain initial conditions (e.g., zero initial velocity, and/or misalignment angle near the hilltop) which allow certain approximations to be made, one has to resort to numerical solutions. 
While numerical solutions are sometimes essential for making contact with data, they obscure the connection to theory. 
It would be desirable to develop analytic tools to infer how the model parameters and the initial conditions affect the cosmological history. 
In this paper, we find an analytic bound that the axion mass parameter and decay constant must fulfill at all times as the universe expands while decreasing its acceleration rate.
This means that the DE scenario to which our bounds apply most naturally are the so-called \emph{thawing quintessence} models \cite{Caldwell:2005tm}; for reviews, see refs.~\cite{Copeland:2006wr, Tsujikawa:2013fta}.
These are EFTs where the state parameter $w_\phi$ of the quintessence field has been increasing over time, despite being still sufficiently close to $w_\phi=-1$ so as to drive cosmic acceleration today.
Our results equip one with analytic control over the phase space\footnote{By ``phase space'' we mean the space of dynamical observables, in the dynamical-system sense that will be laid out below. In practice, we can think of the density parameters $\smash{\Omega_\alpha = \rho_\alpha / \rho_{H}}$ and state parameters $\smash{w_\alpha = p_\alpha/\rho_\alpha}$, where $\smash{\rho_\alpha}$ and $\smash{p_\alpha}$ are the energy density and pressure of a given universe component, and $\rho_{H}$ is the critical energy density set by the Hubble parameter.
By ``parameter space'', also used below, we mean the space of couplings for the axion field.}
whose evolution is consistent with a given set of boundary conditions.
We emphasize that our results hold independently of the initial conditions: in particular, we do not need to resort to slow-roll analyses nor do we need to assume the initial field velocities and misalignment angles to be small.
In fact, the DESI DR1 data \cite{DESI:2024mwx} have been shown to favor the scenario where, starting displaced from the potential maximum, the axion initially climbs up the potential before reaching a turning point and subsequently rolling down \cite{Berbig:2024aee}.
More generally, the preferred initial displacement is correlated with the hilltop curvature and depends on the data set and adopted priors \cite{Bhattacharya:2024kxp}.

For instance, fixing the axion couplings and the measured energy densities today, as in refs.~\cite{DESI:2025fii, Luu:2025fgw, Urena-Lopez:2025rad, Lin:2025gne}, we are able to analytically identify a restricted phase-space subregion that might have led to those currently-measured values.
We found a bound of the form
\begin{equation}
   \Omega_\M^{l} - a \geq - (b + c \, w_\phi \Omega_\phi) \geq 0,
\end{equation}
where $\smash{\Omega_\M}$ and $\Omega_\phi$ represent the density parameters associated with matter and the axion field, respectively, as exemplified in figs.~\ref{fig.: DESI bounds} and \ref{fig.: DES bound}.
The numerical values of the coefficients $a$, $b$ and $c$ and of the power $l$ depend on the data sets that are taken into account, as in eqs.~(\ref{master bound: DESI+Pantheon+}, \ref{master bound: DESI+Union3}, \ref{master bound: DESI+DES}) and (\ref{AdS master bound: DES OmegaM}, \ref{AdS master bound: DES OmegaLambda}).
Clearly, relationships of this kind find a natural application in narrowing down the model and/or parameter space that are constrained using numerical approaches.

The reverse logic can also be applied, fixing our position in phase space at some point in the past and today -- e.g. fixing the Hubble and equation-of-state parameters -- to narrow down the allowed scalar couplings.
As an application, we make use of our analytic bounds to quantify the implications of a quantum gravity expectation known as the axionic weak gravity conjecture (AWGC) \cite{Arkani-Hamed:2006emk, Rudelius:2015xta, Brown:2015iha, Bachlechner:2015qja, Heidenreich:2015wga}; see refs.~\cite{Harlow:2022ich, Montero:2024qml} for a review.
Besides subplanckian axion decay constants, which have long been known to be implied by the AWGC, we find a lower bound on the axion mass in order to be compatible with the current universe, i.e.
\begin{equation}
    \m \geq O(10^2) H_\fin,
\end{equation}
where $H_\fin$ is the Hubble parameter today, as in eq.~(\ref{AWGC: mass lower bound}) and exemplified by fig.~\ref{fig.: (m,f)-parameter space bound}.
As we will discuss, this is in tension with the intuition of masses of order $m \leq O(1) H_0$ for the quintessence field \cite{Cicoli:2023opf}.
Finally, as a spinoff of our current work, we revisit a previous argument by Kamionkowski, Pradler and Walker \cite{Kamionkowski:2014zda}, and find a different parametric scaling for the likelihood that DE today emerges from a random distribution of initial conditions in the string axiverse.

\section{Cosmic axions}

Let the $d$-dimensional spacetime be described by the FLRW metric
\begin{equation} \label{FLRW-metric}
    d s_{1,d-1}^2 = - \de t^2 + a^2(t) \, \de l_{d-1}^2,
\end{equation}
where $a$ is the scale factor, with the Hubble parameter $\smash{H = \dot{a}/a}$.
While $d=4$ is our focus for cosmological applications, we keep an arbitrary $d$ in our general analysis, and we express the reduced Planck mass as $\smash{m_{\p, d} = \kappa_d^{-2/(d-2)}}$.
We consider a canonical field $\phi$, i.e. with kinetic term $T = \dot{\phi}^2/2$, subject to a scalar potential $V$, and a set of cosmological fluids of energy density $\rho_\alpha$ and pressures $p_\alpha = w_\alpha \rho_\alpha$, with $w_\alpha \in \; [-1,+1]$.
Then, the cosmological equations can be written as
\begin{subequations}
\begin{align}
    & \ddot{\phi} + (d-1) H \dot{\phi} + V' = 0, \label{FRW-KG eq.} \\[1.75ex]
    & \dot{\rho}_\alpha + (d-1) (1+w_\alpha) H \rho_\alpha = 0, \label{continuity eq.} \\
    & H^2 = \dfrac{2 \kappa_d^2}{(d-1) (d-2)} \biggl[ \dfrac{1}{2} \, \dot{\phi}^2 + V + \sum_\alpha \rho_\alpha \biggr], \label{Friedmann eq.}
\end{align}
\end{subequations}
where all variables are assumed to only depend on the cosmological time.
A combination of eqs.~(\ref{FRW-KG eq.}, \ref{continuity eq.}, \ref{Friedmann eq.}) further gives $-\dot{H} = \kappa_d^2 [\dot{\phi}^2 + \sum_\alpha (1+w_\alpha) \rho_\alpha]/(d-2)$.
We stress that we can study both flat ($k=0$) and spatially-hyperbolic ($k=-1/\ell^2)$ FLRW models: in the latter case, the curvature terms can be formally treated as a perfect fluid with state parameter $w = -(d-3)/(d-1)$.
Similarly, a cosmological constant would appear with $w = -1$.

As a working axion model, we will initially consider the potential
\begin{equation} \label{V}
    V = \m^2 \f^2 \, \Bigl[ 1 - \cos \, \Bigl(\dfrac{\phi}{\f}\Bigr) \Bigr],
\end{equation}
where $\f$ is the axion decay constant, with mass dimension $[\f] = [\phi] = d/2-1$, and with a $2 \pi \f$-periodicity for the field, and $\m$ is the axion mass parameter.
The potential in eq.~(\ref{V}) emerges ubiquitously in string compactifications \cite{Svrcek:2006yi, Arvanitaki:2009fg, Cicoli:2012sz, Marsh:2015xka}; see also ref.~\cite{Anchordoqui:2025fgz}.
In 4d models, it provided the original model of ``natural inflation'' \cite{Freese:1990rb, Adams:1992bn, Frieman:1995pm, Pajer:2013fsa}.
The same potential is also widely discussed as a model of DE; see e.g. ref.~\cite{Copeland:2006wr}.

If we assume the initial conditions to be such that $\smash{\dot{\phi}_\In = 0}$ and such that the Hubble scale is very large, then the axion is effectively frozen by Hubble friction.
Here, ``initial'' refers to the beginning $t_b$ of the cosmological
time interval under consideration; a formal definition will be specified
in sec.~\ref{ssec.: analytic bound}.
After $H$ has gone down enough, there are two opposite limiting regimes:
\begin{enumerate*}[label=(\roman*)]
    \item if $\phi_\In \simeq 0$, the axion oscillates around the minimum, with the scalar fluid behaving like matter, before eventually freezing;
    \item if $\phi_\In \simeq \f \pi$, the axion field slowly rolls towards the minimum, with a temporary effective parameter $w_\phi < -1/3$, in which case the axion energy density may come to dominate and source a period of cosmic acceleration, before eventually moving towards the minimum.
\end{enumerate*}
Yet, the axion may as well be in an intermediate regime.
This is the most likely scenario if the initial positions are distributed uniformly.
However, making quantitative predictions in such a regime is much harder.
Furthermore, one may be interested in situations where the initial axionic kinetic energy is not negligible.

In this note, we derive analytic bounds that allow one to control the dynamical evolution of the cosmological observables across generic phases of axion field displacement, for generic initial conditions.
Our results will in fact apply to generic potentials of the form
\begin{equation} \label{generic V}
    V_{\mathrm{tot}} = \Lambda + \m^2 \f^2 \, \Bigl[ 1 - \cos \, \Bigl(\dfrac{\phi}{\f}\Bigr) \Bigr]^p,
\end{equation}
for arbitrary values of the vacuum energy $\Lambda > -2^p \m^2 \f^2$, of either sign, and of the power $p > 1$.
For an illustration, see fig.~\ref{fig.: axion potential}.
We will also generalize the conclusions to multi-field scenarios.

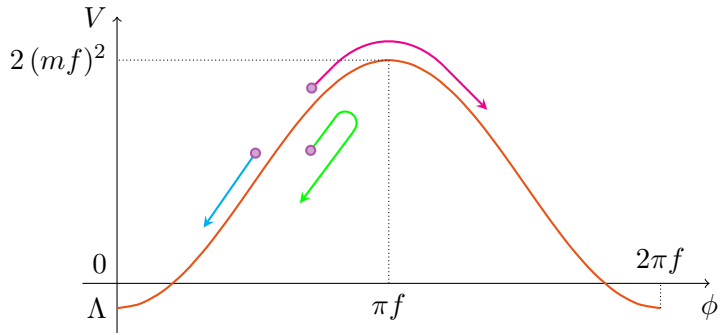
\begin{figure}[h]
    \centering
    
    \begin{tikzpicture}[xscale=1.15,yscale=1.65,every node/.style={font=\normalsize},ball/.style={draw,circle,minimum size=1.2mm,inner sep=0pt,outer sep=0pt,violet!65!white,solid,fill=violet!35!white}]

    \draw[->] (-0.4,0) -- (6.85,0) node[below]{$\phi$};
    \draw[->] (0,-0.4) -- (0,2.15) node[left]{$V_{\mathrm{tot}}$};

    \draw[domain=0:2*pi, smooth, thick, variable=\x, orange!65!purple] plot ({\x}, {0.8-cos(deg(\x))}) node[above right, black]{};

    \node[above left] at (0,0){$0$};
    \node[left] at (0,-0.2){$\Lambda$};

    \draw[densely dotted] (pi,0) node[below]{$\pi \f$} -- (pi,1.8) -- (0,1.8) node[left]{$\Lambda + 2 \, (\m \f)^2$};
    
    \draw[densely dotted, above] (2*pi,-0.2) -- (2*pi,0) node[above]{$2\pi \f$};

    \draw[thick, magenta] (2.25,1.575) node[ball]{} -- (2.5,1.75);
    \draw[thick, magenta] (2.5,1.75) arc[start angle=125, end angle=55,radius=1.116];
    \draw[thick, magenta, -stealth] (3.78,1.75) -- (4.28,1.40);

    \draw[thick, green,rotate=-2] (2.2,1.15) node[ball]{} -- (2.50,1.45) arc[start angle=135, end angle=-45, x radius=0.15/1.15, y radius=0.15/1.65];
    \draw[thick, green, -stealth,rotate=-2] (2.69,1.32) -- (2.09,0.72);
    
    \draw[thick, cyan, -stealth] (1.6,1.05) node[ball]{} -- (1.0,0.45);

    \end{tikzpicture}
    
    \caption{The scalar potential $\smash{V_{\mathrm{tot}} = \Lambda + \m^2 \f^2 \, \bigl[ 1 - \cos \, (\phi/\f) \bigr]}$, with $\Lambda<0$ (orange), along with illustrative ``field trajectories'': a monotonic fall (cyan), a fall after a motion inversion (green), and a hilltop-crossing (magenta).}
    
    \label{fig.: axion potential}
\end{figure}

\section{Dynamics of transient epochs}
In this section, we compute dynamically a lower bound on the duration of the acceleration phase for models of axion-driven cosmic acceleration.

Our results are stated in terms of the parameter $\epsilon = - \dot{H}/H^2$, which measures the variation of the Hubble factor accounting for all the components of the universe. This $\epsilon$ parameter is related to the deceleration parameter $q = - a \ddot{a} / \dot{a}^2$ as $\epsilon = 1 + q$.
On the other hand, the first slow-roll parameter $\smash{\epsilon_V = [(d-2)/(4 \kappa_d^2)] (\ab \mathrm{grad} \, V \ab /V)^2}$ represents a measure of the steepness of the potential, but a priori, it does not capture any dynamics.

\subsection{Dynamical system}
To start, let $\Lambda=0$ and $p=1$.
A fruitful way to find solutions employs the formulation of the equations as an autonomous dynamical system of ordinary differential equations (ODEs).
This formulation was originally introduced to solve cosmological solutions with exponential potentials \cite{Copeland:1997et}, where one can exploit the proportionality between $V$ and $V'$. 
Here, we show how to apply a similar logic to a periodic potential, which after all can be obtained via the analytic continuation of an exponential.

Given
\begin{equation}
    \gamma = \dfrac{1}{\kappa_d \f},
\end{equation}
let
\begin{subequations}
\begin{align}
    x & = \dfrac{\kappa_d}{\sqrt{d-1} \sqrt{d-2}} \, \dfrac{\dot{\phi}}{H}, \label{x} \\
    \yr & = - \dfrac{\kappa_d \, \m \f}{\sqrt{d-1} \sqrt{d-2}} \, \dfrac{1}{H} \, \sin \, \Bigl(\dfrac{\kappa_d}{2} \gamma \phi\Bigr), \label{yr} \\
    \yi & = + \dfrac{\kappa_d \, \m \f}{\sqrt{d-1} \sqrt{d-2}} \, \dfrac{1}{H} \, \cos \, \Bigl(\dfrac{\kappa_d}{2} \gamma \phi\Bigr), \label{yi} \\
    z_\alpha & = \dfrac{\kappa_d \sqrt{2}}{\sqrt{d-1} \sqrt{d-2}} \, \dfrac{\sqrt{\rho_\alpha}}{H}, \label{z} \\[1.4ex]
    h & = (d-1) H, \label{h}
\end{align}
\end{subequations}
as well as
\begin{subequations}
\begin{align}
    c & = \dfrac{\gamma}{\Gamma_d} = \dfrac{1}{2} \, \dfrac{\sqrt{d-2}}{\sqrt{d-1}} \, \dfrac{1}{\kappa_d \f}, \label{c} \\
    l & = \dfrac{\m \Gamma_d}{\sqrt{2} \gamma} = \sqrt{2} \, \dfrac{\sqrt{d-1}}{\sqrt{d-2}} \, \kappa_d \m \f, \label{l}
\end{align}
\end{subequations}
with $\Gamma_d = 2 \sqrt{d-1}/\sqrt{d-2}$.
The functions $x$ and $z_\alpha$ represent the axionic kinetic and fluid energy densities, respectively, whereas $\yi$ and $\yr$ parameterize the potential and its derivatives, all in appropriate Hubble units.
Then, eqs.~(\ref{FRW-KG eq.}, \ref{continuity eq.}, \ref{Friedmann eq.}) are equivalent to the autonomous dynamical system of ordinary differential equations
\begin{subequations}
\begin{align}
    \dot{x} & = \biggl[ -x + 4 c \, \yr \yi + x \, \Bigl[ (x)^2 + \sum_\beta \dfrac{1+w_\beta}{2} \, (z_\beta)^2 \Bigr] \biggr] \, h, \label{x-equation} \\
    \yrd & = \biggl[ \Bigl[ (x)^2 + \sum_\beta \dfrac{1+w_\beta}{2} \, (z_\beta)^2 \Bigr] \yr - c x \yi \biggr] \, h, \label{yr-equation} \\[0.45ex]
    \yid & = \biggl[ \Bigl[ (x)^2 + \sum_\beta \dfrac{1+w_\beta}{2} \, (z_\beta)^2 \Bigr] \yi + c x \yr \biggr] \, h, \label{yi-equation} \\[0.45ex]
    \dot{z}_\alpha & = \biggl[ - \dfrac{1+w_\alpha}{2} + (x)^2 + \sum_\beta \dfrac{1+w_\beta}{2} \, (z_\beta)^2 \biggr] \, z_\alpha h, \label{z-equation} \\
    \dot{h} & = - \biggl[ (x)^2 + \sum_\beta \dfrac{1+w_\beta}{2} \, (z_\beta)^2 \biggr] h^2, \label{h-equation}
\end{align}
\end{subequations}
subject to the constraints
\begin{subequations}
    \begin{align}
        & (x)^2 + 4 (\yr)^2 + \sum_\beta (z_\beta)^2 = 1, \label{sphere condition} \\
        & (\yr)^2 + (\yi)^2 = \dfrac{1}{2} \dfrac{l^2}{h^2}. \label{(yr,yi)-bounds}
    \end{align}
\end{subequations}
This choice of variables is the most efficient one for our study, and its relationship with other variable choices is in app. \ref{app: alternative coordinates for the dynamical system}.
The $\epsilon$-parameter appears explicitly in the ODEs via $\xi = - \dot{h}/h^2 = \epsilon/(d-1)$.
The simplicity of the formulation above is clear: for instance, via eqs.~(\ref{h-equation}, \ref{sphere condition}) we immediately see that $\smash{\xi < (x)^2 + \sum_\beta (z_\beta)^2 \leq 1}$, i.e. $\epsilon \leq d-1$.

Below, we will prove a dynamical bound on $\xi$, and hence on $\epsilon$.
For simplicity, the discussion will be initially referred to the case of a single perfect fluid $\rho$: with it representing matter, and the field sourcing DE, this is an economic and yet non-trivial model of the current universe.
All results generalize straightforwardly in the presence of multiple fluids.

\subsection{Analytic bound} \label{ssec.: analytic bound}

Our goal is to compute the duration of a cosmological epoch in which the $\epsilon$-parameter has evolved from a given initial value $\epsilon_\In$ to a final value $\epsilon_\fin$.
For instance, this may be the time elapsed between matter/DE equivalence and today.

Formally, let $t_\In$ be an initial time such that $\xi_\In = \xi(t_\In)$, and let $\xi_\fin$ be an arbitrary value $\xi_\fin \geq \xi_\In$.
Then, let $t_\fin$ be the earliest time $t \geq t_\In$ such that $\xi(t)$ has grown up to the value $\xi_\fin$; in other words, $\smash{t_\fin}$ represents the earliest time such that $\smash{\xi(t) \leq \xi_\fin}$ at all times $t \in \; [t_\In, t_\fin]$.
For a fixed $\xi_\fin$, our results will be valid for the time interval $[t_\In, t_\fin]$, of duration $\Delta t_{\In \fin} = t_\fin - t_\In$.
For the time being, we will refer to the elapsed time; later, we will express the bound in terms of cosmological observables such as the density parameters.
Because the dynamical system is formulated naturally in the reduced variables in eqs.~(\ref{x}, \ref{yr}, \ref{yi}, \ref{z}, \ref{h}) and eqs.~(\ref{c}, \ref{l}), so will be the formal results in their bare derivation: this makes it immediate to follow the mathematical reasoning. Right after the derivation, the results will be formulated in terms of physical variables, accompanied by a physical interpretation.
 
\subsubsection{Universal bound}
Here we derive a universal bound on the minimal duration of the time interval $[t_\In, t_\fin]$.

\begin{lemma} \label{lemma: abx upper bound}
    One has
    \begin{equation} \label{eq.: abx upper bound}
        \ab x(t) \ab \leq \ab x_\In \ab + \sqrt{2} \, c l (t-t_\In).
    \end{equation}
\end{lemma}

\begin{proof}
    Looking at eq.~(\ref{x-equation}), we may describe the time evolution of $\smash{\ab x \ab = \sqrt{(x)^2}}$ as
    \begin{align*}
        \dfrac{\de}{\de t} \, \ab x \ab = \bigl[ - \ab x \ab (1- \xi) + 4 \mathrm{sgn}(x) \, c \, \yr \yi \bigr] \, h.
    \end{align*}
    Because $\xi \leq 1$, we may write $\smash{(\de / \de t) \, \ab x \ab \leq 4 \mathrm{sgn}(x) \, c \, \yr \yi \, h}$.
    In view of eq.~(\ref{sphere condition}), we know that $\smash{\ab \yr \ab \leq 1/2}$, and, in view of eq.~(\ref{(yr,yi)-bounds}), we can write $\smash{\ab \yi \ab h \leq l/\sqrt{2}}$.\footnote{We highlight that the bound becomes optimal at late times.
    This is because $\ab \yr \ab \leq 1/2$ by eq.~(\ref{sphere condition}) and therefore, as long as $\lim_{t \to \infty} h = 0$, eq.~(\ref{(yr,yi)-bounds}) actually implies $\smash{\lim_{t \to \infty} \ab \yi \ab h = l/\sqrt{2}}$.}
    Hence, we may write $\smash{(\de / \de t) \ab x \ab \leq \sqrt{2} \, c l}$, whence eq.~(\ref{eq.: abx upper bound}).
\end{proof}

From now on, let $\xi_\fin \leq (1+w)/2$.
Physically, this is a totally safe assumption in a DE model, with $d=4$: in a minimal setting with matter ($w=0$) and the DE field, a phase of cosmic acceleration like the one the universe has been in the past $\e$-fold or so certainly requires $\epsilon = 3 \xi < 1$, consistently with $\epsilon \leq \epsilon_\fin < 3/2$.

\begin{lemma} \label{lemma: time-interval lower bound}
    The duration of the time interval $[t_\In, t_\fin]$ is bounded from below as
    \begin{equation} \label{eq.: time-interval lower bound}
        \Delta t_{\In \fin} \geq \dfrac{1}{2 \sqrt{2 \xi_\fin}} \dfrac{1}{cl} \, (\xi_\fin - \xi_\In).
    \end{equation}
\end{lemma}

\begin{proof}
    Let us split the proof into two main steps.
    \begin{enumerate}[label=\roman*.,wide, labelwidth=!, labelindent=0pt, parsep=0pt, itemsep=0pt,leftmargin=12pt]
        \item In view of eq.~(\ref{z-equation}), one may write the differential inequalities $\smash{\dot{z} \leq [-(1+w)/2 + \xi_\fin] \, z h \leq 0}$.
        Then, because the function $z$ is non-increasing, one may write
        \begin{align*}
            (x)^2 - (x_\In)^2 = \xi - \xi_\In - \dfrac{1+w}{2} \, \bigl[ (z)^2 - (z_\In)^2 \bigr] \geq \xi - \xi_\In.
        \end{align*}
        \item After writing $\smash{(x)^2 - (x_\In)^2 = [\ab x \ab + \ab x_\In \ab] [\ab x \ab - \ab x_\In \ab]}$, in view of eq.~(\ref{eq.: abx upper bound}) we deduce
        \begin{align*}
            (x)^2 - (x_\In)^2 & \leq [\ab x \ab + \ab x_\In \ab] \sqrt{2} \, c l (t-t_\In) \leq 2 \sqrt{2 \xi_\fin} \, cl (t-t_\In),
        \end{align*}
        where we also took advantage of the inequality $\smash{(x)^2 \leq \xi_\fin - [(1+w)/2] \, (z)^2 \leq \xi_\fin}$.
    \end{enumerate}
    By combining the two bounds above, it is then immediate to infer the bound in eq.~(\ref{eq.: time-interval lower bound}).
\end{proof}

The calculation can be improved:
\begin{enumerate*}[label=(\alph*)]
    \item in constraining $\ab x \ab$, we may keep the linear term in $\ab x \ab$;
    \item in constraining $z$, we can constrain the functional form of $z$, beyond just observing that it decreases.
\end{enumerate*}
We relegate these technical advances to app. \ref{app.: refined bounds}: see in particular eq.~(\ref{eq.: improved time-interval lower bound}).

Let us now refine the bound with a simple observation.
With the misalignment angle $\delta = \pi - \phi/\f$, and writing $\smash{\yi h = (l/\sqrt{2}) \, \sin \, \delta/2}$, we may parameterize the maximum misalignment achieved within the interval $[t_\In, t_\fin]$ through the identification $\smash{b = \max_{t \in\; [t_\In, t_\fin]} \ab \sin \, [\delta(t)/2] \ab}$.
As a consequence, we may write $\smash{\ab \yi \ab h \leq b l /\sqrt{2}}$.

\begin{lemma} \label{lemma: refined time-interval lower bound}
    The duration of the time interval is bounded from below as
    \begin{equation} \label{eq.: refined time-interval lower bound}
        \Delta t_{\In \fin} \geq \dfrac{1}{2 \sqrt{2 \xi_\fin}} \dfrac{1}{b cl} \, (\xi_\fin - \xi_\In).
    \end{equation}
\end{lemma}

\begin{proof}
    The proof proceeds exactly as for eqs.~(\ref{eq.: abx upper bound}, \ref{eq.: time-interval lower bound}), with the improved parameterization of the range of $\yi$.
\end{proof}

\subsubsection{Physical units}

In physical units, we have $c l = \m/\sqrt{2}$.
If we  define $\smash{C_{\In \fin} = \max_{t \in\; [t_\In, t_\fin]} \ab \cos \, [\phi(t)/(2\f)] \, \ab}$ -- which is just a measure of the maximum displacement from the hilltop throughout the dynamical evolution from $t_\In$ to $t_\fin$ --, then we can simply express eq.~(\ref{eq.: abx upper bound}) as
\begin{equation}
    \dfrac{\kappa_d \ab \dot{\phi} \ab}{H} \leq \dfrac{\kappa_d \ab \dot{\phi}_\In \ab}{H_\In} + \sqrt{d-1} \sqrt{d-2} \, m \Delta t_{\In \fin},
\end{equation}
namely a bound on the field velocity in Hubble units.
In addition, eq.~(\ref{eq.: refined time-interval lower bound}) corresponds to
\begin{equation} \label{refined physical time-interval lower bound}
    \Delta t_{\In \fin} \geq \dfrac{1}{2 m} \, \dfrac{1}{\sqrt{d-1} \sqrt{\epsilon_\fin}} \, \dfrac{\epsilon_\fin - \epsilon_\In}{C_{\In \fin}}.
\end{equation}
This is a quantitative relationship between the duration of a given accelerating phase and the displacement angle.
It applies not only to the hilltop region, but to all regions in the $\phi$-domain, irrespectively of the value of $\epsilon_V$, and it requires no assumptions on the initial conditions.
In cases where $\ab \phi_\In/\f - \pi \ab \leq \ab \phi_\fin/\f - \pi \ab$, i.e. the initial angle is closer to the hilltop than the final angle, we may identify $\smash{C_{\In \fin} = \ab \cos \, [\phi_\fin/(2\f)] \, \ab}$.
Examples are monotonic falls, falls following a motion inversion, and hilltop-crossings, such as in fig.~\ref{fig.: axion potential}; ``field trajectories'' different from those displayed in fig.~\ref{fig.: axion potential} are also possible, just with different values for $C_{\In \fin}$.

\subsection{Generalizations}
Here we describe a few immediate generalizations, which come up as natural extensions of the basic model above both in numerical data fits and theoretical realizations.

\subsubsection{Potentials with dS or AdS shifts}

We can include a cosmological constant $\Lambda$.
If $\Lambda > 0$, the model hosts a stable de Sitter (dS) vacuum; if $\Lambda < 0$, the model instead hosts a stable AdS vacuum.\footnote{Of course, even if $\Lambda <0$, there can exist a phase of cosmic acceleration if $\ab \Lambda \ab \leq 2 \m^2 \f^2$ and the initial conditions are such that $V_{\Lambda,\In} > 0$.}

In the presence of a cosmological constant $\Lambda$, the dynamical system is modified straightforwardly.
Indeed, eqs.~(\ref{x-equation}, \ref{yr-equation}, \ref{yi-equation}, \ref{z-equation}, \ref{h-equation}) are unchanged.
Then, given
\begin{equation} \label{L}
    L^2 = \dfrac{d-1}{d-2} \, 2 \smash{\kappa_d^2} \Lambda,
\end{equation}
the geometric constraint in eq.~(\ref{sphere condition}) gets modified to
\begin{equation} \label{sphere condition with c.c.}
    (x)^2 + 4 (\yr)^2 + (z)^2 = 1 - \dfrac{L^2}{h^2}.
\end{equation}
Within our derivation method, this affects our estimate on the range of $\yr$.
In particular, the variable $\yr$ is now bounded as $\smash{\ab \yr \ab \leq \sqrt{1 - L^2/h^2}/2}$.
Because $h$ is a non-increasing function, we can proceed as follows, based on the sign of $L^2$.
\begin{itemize}[wide, labelwidth=!, labelindent=0pt, parsep=4pt, itemsep=0pt,leftmargin=12pt]
    \item With an AdS shift, where $L^2<0$, we get $\smash{\ab \yr \ab \leq \sqrt{1 + \ab L^2 \ab/h^2}/2 \leq \sqrt{1 + \ab L^2 \ab/h_\fin^2}/2}$.
    Then, the bound in eq.~(\ref{eq.: refined time-interval lower bound}) gets modified to
    \begin{equation} \label{AdS-models: time-interval lower bound}
        \Delta t_{\In \fin} \geq \dfrac{1}{2 \sqrt{2 \xi_\fin}} \dfrac{1}{b cl} \, \dfrac{\xi_\fin - \xi_\In}{\sqrt{1 + \ab L^2 \ab / h_\fin^2}}.
    \end{equation}
    \item With a dS shift, where $L^2>0$, we get $\smash{\ab \yr \ab \leq \sqrt{1 - \ab L^2 \ab/h^2}/2 \leq \sqrt{1 - \ab L^2 \ab/\smash{h_\In^2}}/2}$.
    Then, the bound in eq.~(\ref{eq.: refined time-interval lower bound}) gets modified to
    \begin{equation} \label{dS-models: time-interval lower bound}
        \Delta t_{\In \fin} \geq \dfrac{1}{2 \sqrt{2 \xi_\fin}} \dfrac{1}{b cl} \, \dfrac{\xi_\fin - \xi_\In}{\sqrt{1 - \ab L^2 \ab/\smash{h_\In^2}}}.
    \end{equation}
    As a consistency check, we note that the time interval is infinite if e.g. $h_\In = L$.
    This is correct because, following eq. (\ref{sphere condition with c.c.}), this initial condition is the infinitely-fine tuned scenario with the field sitting still at a minimum, with zero fluid energy density, thereby realizing classically an eternal dS solution.
\end{itemize}
Coherently with physical intuition, a negative cosmological constant lowers the minimal time-interval length, while a positive one achieves the opposite effect.

All results will be expressed in physical units below, after the further generalization to multi-axion scenarios.

\subsubsection{Multiple axions}
Let us consider a multifield generalization of the model we discussed before, i.e. $n$ canonical axions $\phi^a$, for $a=1,\dots,n$, with a total potential $\smash{V = \sum_{a=1}^n (\m_a \f^a)^2 [1 - \cos \, (\phi^a/\f^a)]}$.

In this case, the dynamical system is generalized straightforwardly, with $n$ variables $x^a$, $\smash{\yr^a}$, $\smash{\yi^a}$ and $2 n$ parameters $c^a$ and $l^a$.
It is clear that the bound in eq.~(\ref{refined physical time-interval lower bound}) is modified immediately via the substitution $\m \mapsto \sum_{a=1}^n m_a$.

As should be clear by now, a combination of the results above is also simple.
For brevity, let $\smash{\Omega_\Lambda = 2 \kappa_d^2 \Lambda / [(d-1)(d-2) H^2]}$.
Then, in the case of an axiverse with an eventual AdS minimum, one may write
\begin{equation} \label{axiverse AdS-models: physical time-interval lower bound}
    \Delta t_{\In \fin} \geq \dfrac{1}{2 M_n} \, \dfrac{1}{\sqrt{d-1} \sqrt{\epsilon_\fin}} \, \dfrac{\epsilon_\fin - \epsilon_\In}{\sqrt{1 + \ab \Omega_{\Lambda,\fin} \ab}},
\end{equation}
while for the case of a dS uplift we get
\begin{equation} \label{axiverse dS-models: physical time-interval lower bound}
    \Delta t_{\In \fin} \geq \dfrac{1}{2 M_n} \, \dfrac{1}{\sqrt{d-1} \sqrt{\epsilon_\fin}} \, \dfrac{\epsilon_\fin - \epsilon_\In}{\sqrt{1 - \Omega_{\Lambda,\In}}},
\end{equation}
where for brevity we defined the effective total mass parameter $\smash{M_n = \sum_{a=1}^n m_a \, C_{\In \fin}^a}$, with $\smash{C_{\In \fin}^a = \max_{t \in\; [t_\In, t_\fin]} \ab \cos \, [\phi^a(t)/(2\f^a)] \, \ab}$.

A generalization to a model with a potential $\smash{V = \sum_{a=1}^n \Lambda_a [1 - \cos\, ( \sum_b (f^{-1})^a{}_b \phi^b + \delta^a)]}$ -- see e.g. ref.~\cite{Katewongveerachart:2026ovj} for a recent study of multi-axion quintessence\footnote{The study in ref.~\cite{Katewongveerachart:2026ovj} focused on quintessence with $n=2$ axions, while generalizations to $n \gg1$ require  fair ensembles of UV complete models. The sampling algorithm developed in ref.~\cite{Yip:2025hon} can potentially provide such ensembles.} -- is also possible via the same methods, but we defer its analysis to future work.

\subsection{Generalized periodic potentials}

Our methods also apply naturally to generalized periodic potentials of the form
\begin{equation} \label{generalized V}
    V_{p} = \m^2 \f^2 \, \Bigl[ 1 - \cos \, \Bigl(\dfrac{\phi}{\f}\Bigr) \Bigr]^p,
\end{equation}
where $p > 1$ is an arbitrary power.
Such potentials play a prominent role in the context of early dark energy (EDE, \cite{Poulin:2018cxd, McDonough:2021pdg}); we will comment on this later.
In any case, one may also study them for DE phenomenology.
Evidently, we identify $V_1 \equiv V$, with $V$ given in eq.~(\ref{V}).
For simplicity, we still refer to $\m$ as the ``mass parameter''.
As detailed in app. \ref{app.: generalized bounds} -- see eq.~(\ref{eq.: generalized time-interval lower bound}) -- a simple generalization of the methods laid out before allows one to find the lower bound
\begin{equation} \label{generalized physical time-interval lower bound}
    \Delta t_{\In \fin} \geq \dfrac{2^{\frac{p-3}{2}} \Omega_{\m \f,\In}^{\frac{p-1}{2p}}}{p \sqrt{d-1} \sqrt{\epsilon_\fin}} \, \dfrac{\epsilon_\fin - \epsilon_\In}{\mu},
\end{equation}
where we defined $\smash{\Omega_{\m \f,\In} = 2^{1+p} \kappa_d^2 \, \m^2 \f^2 / [(d-1)(d-2) H_\In^2]}$ for brevity.
It is apparent that the value $p=1$ simplifies the bound significantly.
All larger values of $p$ enhance the minimal duration of the accelerating phase, coherently with the intuition that the potential is flatter.

\subsection{Bounds and cosmological observables}

We can express our bound in terms of other observables besides $\Delta t_{\In \fin}$, such as the density parameters $\Omega_A = \rho_A/\rho_H$, where $\smash{\rho_H = (d-1)(d-2) H^2/2\kappa_d^2}$, and the equation-of-state parameters $w_A$.
A model with matter, a cosmological constant and the field gives $\smash{\Omega_{\M} = \rho_{\M}/\rho_H}$, $\smash{\Omega_{\Lambda} = \Lambda/\rho_H}$, and $\smash{\Omega_{\phi} = 1 - \Omega_{\M} - \Omega_\Lambda}$, respectively, with the equation-of-state parameters $w_\M = 0$, $w_\Lambda = -1$, and $\smash{w_\phi = [\dot{\phi}^2/2-V(\phi)][\dot{\phi}^2/2+V(\phi)]}$.

For instance, let us consider a matter density parameter $\Omega_\alpha$, with constant $w_\alpha$.
Under the same assumption that has been employed so far -- i.e. $\xi_\In \leq \xi_\fin \leq (1+w)/2$, for a given $w$ --, a simple manipulation of eqs.~(\ref{z-equation}, \ref{h-equation}) leads to (see also corollary \ref{corollary: improved z upper bound})
\begin{equation} \label{physical-time: upper bound}
    \Bigl( \dfrac{\Omega_{\alpha, \In}}{\Omega_{\alpha, \fin}} \Bigr)^{\! \frac{1}{2} \frac{1}{\frac{(d-1)(1+w_\alpha)}{2 \epsilon_\fin} - 1}} - 1 \geq \epsilon_\fin H_\In \Delta t_{\In \fin}.
\end{equation}
Because $H$ is non-increasing, we may also write $H_\In \geq H_\fin$ on the right hand-side.
Additionally, we may as well work in terms of e-folds $\smash{N_{\In, \fin} = \ln \, (a_\fin/a_\In) = \int_{t_\In}^{t_\fin} \de t \, H(t)}$ between two events.
Then, we have $\smash{N_{\In, \fin} \geq H_\fin \Delta t_{\In \fin}}$.
One can then plug in an inequality for $\Delta t_{\In \fin}$ such as eq.~(\ref{refined physical time-interval lower bound}), eqs.~(\ref{axiverse AdS-models: physical time-interval lower bound}, \ref{axiverse dS-models: physical time-interval lower bound}), and eq.~(\ref{generalized physical time-interval lower bound}).

We may also get rid of angle dependencies via the identity $\smash{\epsilon = (d-1) (1 - \Omega_\Lambda + w_\phi \Omega_{\phi})/2}$, which allows us to write
\begin{equation} \label{cosine identity}
    \cos^2 \Bigl[\dfrac{\phi}{2 \f}\Bigr] = 1 - \dfrac{(d-1)(d-2) H^2}{8 \kappa_d^2 \f^2 \m^2} \, (1-w_{\phi}) \Omega_{\phi}.
\end{equation}

\section{Axion quintessence}

An obvious application of our bound is for models of cosmic acceleration, where it applies both to controlling numerical fits \cite{Toomey:2025xyo} and to understanding theoretical constraints.

In a minimal model of the late universe, we must include (baryonic and dark) matter and the DE component, respectively with energy densities $\smash{\rho_{\M}}$ and $\smash{\rho_\phi = \dot{\phi}^2/2 + V(\phi)}$, with $V(\phi)$ in eq.~(\ref{V}), and state parameters $\smash{w_{\M} = 0}$ and $\smash{w_\phi = [\dot{\phi}^2/2-V(\phi)][\dot{\phi}^2/2+V(\phi)]}$.
For generality, we also entertain the possibility of the presence of a cosmological constant $\Lambda$.

\subsection{Data fits and analytic predictions} \label{subsec.: data fits}
To exploit our bound for a theoretical understanding of axion quintessence models, we need to relate the potential parameters and the boundary conditions via the observed evolution between an initial time $t_\In$ and today, i.e. $t_\fin$.
In terms of $H$, $\Omega_{\M}$, $\Omega_\Lambda$ (where for definiteness we consider the AdS model with $\Omega_\Lambda < 0$), $\epsilon(\Omega_\M, w_\phi \Omega_\phi)$ and $w_\phi$, and in view of eqs.~(\ref{physical-time: upper bound}, \ref{cosine identity}), the bound in eq.~(\ref{AdS-models: time-interval lower bound}) may be expressed as
\begin{equation} \label{master bound}
    \Bigl( \dfrac{\Omega_{\M,\In}}{\Omega_{\M,0}} \Bigr)^{\!\frac{1}{2} \frac{1}{\frac{3}{2 \epsilon_\fin}-1}} - 1 \geq \dfrac{\dfrac{1}{2\sqrt{3}} \dfrac{\sqrt{\epsilon_\fin} \, [\epsilon_\fin - \epsilon_\In]}{\sqrt{1 - \Omega_{\Lambda,0}}}}{\sqrt{\dfrac{\m^2}{H_\fin^2} - \dfrac{3}{4} \, \dfrac{m_\p^2}{\f^2} \, (1-w_{\phi,0}) \Omega_{\phi,0}}}.
\end{equation}

\subsubsection{Axion DE and DESI DR2}

Let us consider the basic axion DE model, with $\Lambda = 0$.
Numerical analyses appeared in refs.~\cite{Urena-Lopez:2025rad, DESI:2025fii}, combining DESI DR1 and DR2 \cite{DESI:2024mwx, DESI:2025zpo, DESI:2025zgx}, respectively, with CMB data \cite{Planck:2019nip, Carron:2022eyg, ACT:2023kun} and one among further data sets:
\begin{enumerate*}[label=(\alph*)]
    \item \label{Pantheon+} Pantheon+ \cite{Scolnic:2021amr, Brout:2022vxf},
    \item \label{Union3} Union3 \cite{Rubin:2023jdq},
    \item \label{DES} and DES Y5 \cite{DES:2024jxu}.
\end{enumerate*}
This provides one representative fitting example that is useful to show an application of our bound.
An analysis of the basic axion DE model after DESI DR2 also appears in ref.~\cite{Lin:2025gne}, which we will comment on in sec. \ref{sec.: discussion}.

In ref.~\cite{Urena-Lopez:2025rad}, the outcomes of the analysis are as below (95\% confidence level).

\begin{center}
\begin{tabular}{c||ccc} 
    & \ref{Pantheon+} & \ref{Union3} & \ref{DES} \\ [0.5ex] 
    \hline\hline
    $w_{\phi,0} \vphantom{\dfrac{\big|}{\big|}}$ & $-0.92_{-0.05}^{+0.06}$ & $-0.75_{-0.18}^{+0.22}$ & $-0.84_{-0.08}^{+0.10}$ \\ 
    \hline
    $\Omega_{\M,0} \vphantom{\dfrac{\big|}{\big|}}$ & $0.314_{-0.010}^{+0.011}$ & $0.330_{-0.016}^{+0.018}$ & $0.322_{-0.013}^{+0.014}$ \\
    \hline
    $\dfrac{H_\fin}{10^{-33} \, \mathrm{eV}/\hbar} \vphantom{\dfrac{\big|}{\big|}}$ & $1.43_{-0.02}^{+0.02}$ & $1.39_{-0.04}^{+0.03}$ & $1.41_{-0.03}^{+0.02}$ \\
    \hline\hline
    $\log \, \Bigl(\dfrac{\m}{\mathrm{eV}/c^2}\Bigr) \vphantom{\dfrac{\big|}{\big|}}$ & $-32.69^{+0.18}_{-0.18}$ & $-32.48^{+0.20}_{-0.24}$ & $-32.58^{+0.20}_{-0.21}$ \\
    \hline
    $\log \, \Bigl( \dfrac{\f}{m_\p} \Bigr) \vphantom{\dfrac{\big|}{\big|}}$ & $-0.11^{+0.24}_{-0.22}$ & $-0.33^{+0.31}_{-0.23}$ & $-0.22^{+0.31}_{-0.26}$ \\ [1ex]
\end{tabular}
\end{center}

\noindent
Further data we can infer are as follows:
\begin{enumerate*}
    \item[\ref{Pantheon+}] $\epsilon_{0} = 0.55$ and $(\dot{\phi}^2_\fin/2)/V_\fin = 0.04$;
    \item[\ref{Union3}] $\epsilon_{0} = 0.75$ and $(\dot{\phi}^2_\fin/2)/V_\fin = 0.14$;
    \item[\ref{DES}] $\epsilon_{0} = 0.64$ and $(\dot{\phi}^2_\fin/2)/V_\fin = 0.09$.
\end{enumerate*}
As the endpoints are not in a regime of negligible kinetic energy, going beyond slow-roll-like approximations is necessary.
By plugging such data into eq.~(\ref{master bound}), we get
\begin{subequations}
    \begin{align}
        1.40 \, \Omega_{\M,\In}^{0.29} - 1 \geq - (0.32 + 0.51 \, w_{\phi,\In} \Omega_{\phi,\In}) \geq 0; \label{master bound: DESI+Pantheon+} \\
        1.73 \, \Omega_{\M,\In}^{0.50} - 1 \geq - (0.15 + 0.29 \, w_{\phi,\In} \Omega_{\phi,\In}) \geq 0; \label{master bound: DESI+Union3} \\
        1.53 \, \Omega_{\M,\In}^{0.38} - 1 \geq - (0.21 + 0.37 \, w_{\phi,\In} \Omega_{\phi,\In}) \geq 0. \label{master bound: DESI+DES}
    \end{align}
\end{subequations}
Such equations provide an understanding of the region in the $(w_\phi, \Omega_\M)$-phase space that are compatible with ending up at the present values.
The exclusion plots for the three data sets are represented in fig.~\ref{fig.: DESI bounds}.

\begin{figure}[h]
    \centering
    
    \begin{tikzpicture}[xscale=28,yscale=15,every node/.style={font=\normalsize},decoration={markings, 
    mark= at position 0.5 with {\arrow{stealth}}}]
    
    \begin{scope}
    
        \clip (0.20,-1) rectangle (0.51,-0.72);

        \node[cyan,left] at (0.314,-0.92){$(\Omega_{\M,0},w_{\phi,0})_{\mathrm{a}}$};
        \draw[name path = a1, domain=0.314:0.34, smooth, thick, variable=\x, cyan] plot ({\x}, {(1.331-2.7518*\x^(0.2922))/(1-\x)}) node[above right, black]{};
        \draw[name path = a2, domain=0.314:0.369, smooth, thick, variable=\x, cyan] plot ({\x}, {-0.631/(1-\x)}) node[above right, black]{};
        \tikzfillbetween[of = a1 and a2]{cyan, opacity=0.1};
        
        \node[orange,right] at (0.330,-0.75){$\;\; (\Omega_{\M,0},w_{\phi,0})_{\mathrm{b}}$};
        \draw[name path = b1, domain=0.330:0.36, smooth, thick, variable=\x, orange] plot ({\x}, {(2.937-5.955*\x^(0.4950))/(1-\x)}) node[above right, black]{};
        \draw[name path = b2, domain=0.330:0.498, smooth, thick, variable=\x, orange] plot ({\x}, {-0.503/(1-\x)}) node[above right, black]{};
        \tikzfillbetween[of = b1 and b2]{orange, opacity=0.1};

        \node[magenta,left] at (0.322,-0.84){$(\Omega_{\M,0},w_{\phi,0})_{\mathrm{c}}$};
        \draw[name path = c1, domain=0.322:0.352, smooth, thick, variable=\x, magenta] plot ({\x}, {(2.161-4.191*\x^(0.3779))/(1-\x)}) node[above right, black]{};
        \draw[name path = c2, domain=0.322:0.43, smooth, thick, variable=\x, magenta] plot ({\x}, {-0.570/(1-\x)}) node[above right, black]{};
        \tikzfillbetween[of = c1 and c2]{magenta, opacity=0.1};
    
    \end{scope}

    \draw[->] (0.29,-1) node[left]{$-1$} -- (0.51,-1) node[below]{$\Omega_{\M,\In}$};
    \draw[->] (0.30,-1.02) node[below]{$0.30$} -- (0.30,-0.74) node[left]{$w_{\phi,\In}$};

    \end{tikzpicture}
    
    \caption{
    The shaded regions in between the cyan, orange and magenta curves fulfill the bounds in eqs.~(\ref{master bound: DESI+Pantheon+}, \ref{master bound: DESI+Union3}, \ref{master bound: DESI+DES}), respectively.
    In the axion DE model, the dynamics cannot flow to the current values from other areas.
    In eqs.~(\ref{master bound: DESI+Pantheon+}, \ref{master bound: DESI+Union3}, \ref{master bound: DESI+DES}), the inequalities on the right hand-sides are enforcing the condition of a thawing model, i.e. $\epsilon_\In \leq \epsilon$. To derive our bound, however, only the conditions that $\epsilon_\In \leq \epsilon_\fin$ and $\epsilon \leq \epsilon_\fin$ were assumed. Dropping the additional requirement that $\epsilon_\In \leq \epsilon$ just removes the upper boundaries from the exclusion plots.
    }
    
    \label{fig.: DESI bounds}
\end{figure}
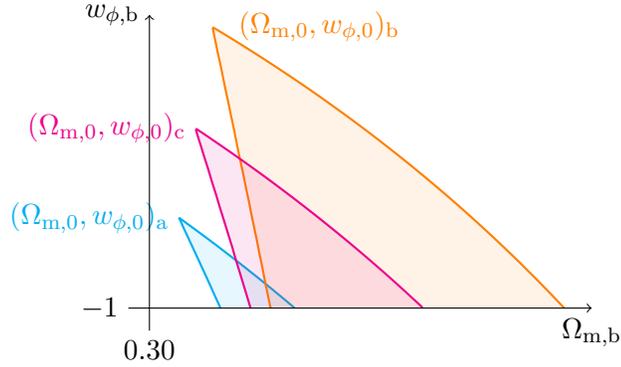

\subsubsection{AdS-shifted axion DE and DES Y6}
It is also instructive to discuss axion DE in the presence of a cosmological constant $\Lambda < 0$.
This scenario intersects with a growing body of interest \cite{Luu:2025fgw, Luu:2025dax}.
Notably,
in ref.~\cite{Luu:2025fgw},
the model was numerically analyzed in
light of the most recent DES Y6 data \cite{DES:2025upx} (see ref.~\cite{Ruchika:2020avj} for an analysis with different data sets).
Under the assumption that $\f=m_\p$, the best-fit values were found to be $\log \, \m c^2 / \mathrm{eV} = -32.533$, $w_{\phi,0} = -0.726$, $\Omega_{\phi,0} = 2.32$, $\Omega_{\Lambda,0}=-1.61$, and $H_\fin/(10^{-33} \, \mathrm{eV}/\hbar) = 1.437$.
With these best-fit values, our eq.~(\ref{master bound}) gives both
\begin{equation} \label{AdS master bound: DES OmegaM}
    2228.46 \Omega_{\M,\In}^{6.23} - 1 \geq 0.29 \, (-0.07 + \Omega_{\Lambda,\In} - w_{\phi,\In} \Omega_{\phi,\In}) \geq 0,
\end{equation}
and
\begin{equation} \label{AdS master bound: DES OmegaLambda}
    \dfrac{1.27}{\sqrt{-\Omega_{\Lambda,\In}}} - 1 \geq 0.29 \, (-0.07 + \Omega_{\Lambda,\In} - w_{\phi,\In} \Omega_{\phi,\In}) \geq 0.
\end{equation}
Clearly, the bounds now involve three unknowns, namely $w_{\phi, \In}$, $\Omega_{\M, \In}$ and $\Omega_{\Lambda, \In}$, but simple phase-space plots can still be drawn, e.g. by fixing a hypersurface with a fixed value of $\Omega_{\Lambda, \In}$ or $\Omega_{\M,\In}$.
An illustrative example is in fig.~\ref{fig.: DES bound}.

\begin{figure}[h]
    \centering
    
    \begin{tikzpicture}[xscale=150,yscale=8.5,every node/.style={font=\normalsize},decoration={markings, 
    mark= at position 0.5 with {\arrow{stealth}}}]

    \begin{scope}
    
        \clip (0.29,-1) rectangle (0.34,-0.45);
    
        \draw[name path = f1, domain=0.29:0.325, smooth, thick, variable=\x, violet!45!white, samples=10] plot ({\x}, {-0.726-23.667*(\x-0.29)-1189.154*(\x-0.29)^2}) node[above right, black]{};
        \draw[name path = f2, domain=0.29:0.325, smooth, thick, variable=\x, violet!45!white] plot ({\x}, {-0.726+7.85*(\x-0.29)}) node[above right, black]{};
        \tikzfillbetween[of = f1 and f2]{violet!45!white, opacity=0.1};
        
    \end{scope}
    
    \draw[densely dotted] (0.29,-1) node[below]{$0.29$} -- (0.29,-0.726) node[left]{$(\Omega_{\M,0},w_{\phi,0})$} -- (0.31,-0.726) node[right]{$-0.726$};

    \draw[->] (0.285,-1) -- (0.31,-1) node[above left]{$-1$} -- (0.33,-1) node[below]{$\Omega_{\M,\In}$};
    \draw[->] (0.31,-1.02) node[below]{$0.31$} -- (0.31,-0.54) node[left]{$w_{\phi,\In}$};

    \end{tikzpicture}
    
    \caption{The shaded region in violet (whose part on the right has been cut off for displaying purposes) fulfills the bound in eq.~(\ref{AdS master bound: DES OmegaM}); to draw this 2d slice of the full 3d plot, the value of $\Omega_{\Lambda,\In}$ has been chosen so as to saturate eq.~(\ref{AdS master bound: DES OmegaLambda}), meaning that $\smash{1.27/\sqrt{-\Omega_{\Lambda,\In}} - 1 = 0.29 \, (-0.07 + \Omega_{\Lambda,\In} - w_{\phi,\In} \Omega_{\phi,\In})}$.}
    
    \label{fig.: DES bound}
\end{figure}
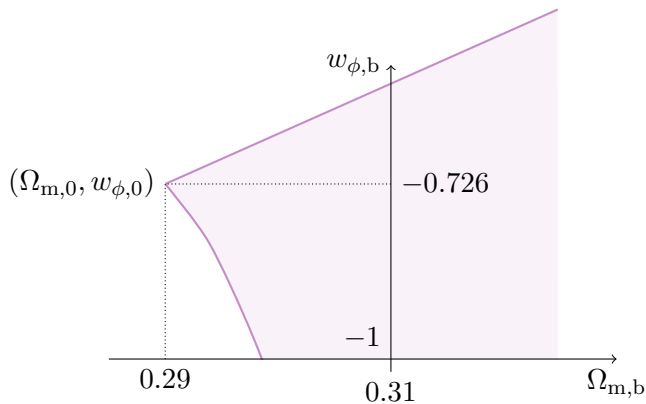

\section{Axionic weak-gravity conjecture and DE}

Our results provide a quantitative tool for assessing the tensions between theoretical expectations and observations.
For definiteness, we will focus on axions with a periodic potential, plus a non-positive cosmological constant $\Lambda \leq 0$ -- although in essence our conclusions are not affected otherwise.
From a higher-dimensional perspective, in the vacuum only the negative spacetime curvature of AdS can be supported by classical sources alone \cite{Maldacena:2000mw, Danielsson:2018ztv}.
In UV complete theories, dS constructions necessitate additional ingredients beyond classical sources (see e.g. ref.~\cite{McAllister:2025qwq} for a recent review on attempts to construct dS vacua in string theory).

To start, in a 4d toy model of the late universe, we assume that DE is represented by a single string axion rolling down its potential, relaxing the $\epsilon$-parameter from $\epsilon_\In$ to the current value $\epsilon_\fin > \epsilon_\In$ over $N_{\In 0}$ e-folds.
Then, based on eq.~(\ref{AdS-models: time-interval lower bound}) and eqs.~(\ref{physical-time: upper bound}, \ref{cosine identity}), we can formulate a bound in the $(\m,\f)$-parameter space as
\begin{equation} \label{axion quintessence: (m,f)-parameter space bound}
    \dfrac{\m^2}{H_\fin^2} \geq \dfrac{3}{2} \dfrac{m_\p^2}{\f^2} \Bigl( 1 - \dfrac{\epsilon_\fin}{3} - \Omega_{\Lambda,0} - \dfrac{\Omega_{\M,0}}{2} \Bigr) + \dfrac{1}{(1-\Omega_{\Lambda,0})} \dfrac{(\epsilon_\fin - \epsilon_\In)^2}{12 \epsilon_\fin N_{\In 0}^2}.
\end{equation}
This same relationship has been discussed above in view of data fits.
Below, we will assess its implications from a theoretical UV perspective.

In a UV-complete EFT of gravity, given an axion $\phi^a$ with decay constant $\f_a$, the AWGC posits that there must exist an instanton of instanton number $k$ such that the instanton action $S_a$ is bounded as
\begin{equation} \label{axionic WGC}
    \dfrac{\kappa_d \f_a}{k} S_a \leq c,
\end{equation}
where $c$ is an order-1 constant.
Evidence was found in refs.~\cite{Andriolo:2020lul,Andriolo:2022rxc} that this order-1 constant $c$ is set by the action-to-charge ratio of Euclidean wormholes.
In a 4d EFT, the axion mass term $\m_a$ is given by
\begin{equation} \label{axion mass}
    \m_a^2 = \dfrac{\Lambda_{\mathrm{UV}}^4 \, \e^{-S_a}}{(\f_a/k)^2},
\end{equation}
where $\Lambda_{\mathrm{UV}} \leq m_\p$ is a UV scale related to the mechanism generating the axion potential \cite{Svrcek:2006yi}, and it can be generally expressed as $\Lambda_{\mathrm{UV}}^4 = m_{\mathrm{S}}^{n_{\mathrm{S}}} m_s^{n_s} m_\p^{4-n_{\mathrm{S}}-n_s}$, for two non-negative powers $n_{\mathrm{S}}, n_s \geq 0$, where $m_{\mathrm{S}}$ is the SUSY-breaking scale and $m_s$ is the string scale \cite{Blumenhagen:2009qh, Arvanitaki:2009fg, Hui:2016ltb, Reece:2025thc}.

If we identify our axion as $\phi = \phi^a$, with $\m = \m_a$ and $\f = \f_a/k$, the statement of eq.~(\ref{axionic WGC}) translates immediately into the bound $\smash{\m^2 \geq \Lambda_{\mathrm{UV}}^4/\f^2 \, \e^{-c m_\p/\f}}$ on the $(\m,\f)$-parameter space.
For convenience, we express the latter as
\begin{equation} \label{axionic WGC: (m,f)-parameter space bound}
    \dfrac{\m^2}{H_\fin^2} \geq \dfrac{\Lambda_{\mathrm{UV}}^4}{H_\fin^2 m_\p^2} \dfrac{m_\p^2}{\f^2} \, \e^{-c \frac{m_\p}{\f}}.
\end{equation}
Due to the tremendous hierarchy between the observed Hubble energy density and the Planck scale, in principle the ratio $\m/H_0$ appears to be gigantic.
It is only thanks to the exponential dependence in the factor $\smash{\e^{-c m_\p/\f}}$ that, potentially, tiny values of $\f/m_\p$ can allow for values $m/H_\fin$ that are not several orders of magnitude larger than unity.
For instance, the parameters $\m$ and $\f$ from data fits for (AdS-shifted) axion DE \cite{Urena-Lopez:2025rad, DESI:2025fii, Luu:2025fgw} discussed in subsec. \ref{subsec.: data fits} violate this bound by several orders of magnitude.\footnote{Indeed, because $\smash{H_\fin/m_\p = O(10^{-60})}$ \cite{Planck:2018vyg}, with axion decay constant such that $\f/m_\p = O(1/10)$ and assuming $c=O(1)$, eq.~(\ref{axionic WGC: (m,f)-parameter space bound}) would at best require $\smash{m/H_\fin \geq 10^{60} \, O(\Lambda_{\mathrm{UV}}/m_\p)^2}$.}
Indeed, the AWGC requires subplanckian values of $\f$.
However, it is not so restrictive on $\m$ if $\f$ is very small.
As it turns out, our dynamical bound quantifies a restrictive lower bound on $\m$ at small $\f$ instead.
The combination of our dynamical bound with the AWGC has important consequences.

\subsection{A lower bound on the axion mass}

Noticeably, the dynamical bound in eq.~(\ref{axion quintessence: (m,f)-parameter space bound}) -- after plugging in the Hubble parameter -- and the AWGC statement in eq.~(\ref{axionic WGC: (m,f)-parameter space bound}) rule out parts of parameter space in a highly non-redundant way: only a specific intersection area is compatible with both.
In particular, their combination singles out a lower bound on the axion mass, as exemplified in fig.~\ref{fig.: (m,f)-parameter space bound}.
This is manifest in log-variables: eqs.~(\ref{axion quintessence: (m,f)-parameter space bound}, \ref{axionic WGC: (m,f)-parameter space bound}) can indeed be expressed respectively as
\begin{align*}
    \ln \, \dfrac{\m}{H_\fin} & \geq - \ln \, \dfrac{\f}{m_\p} + \dfrac{1}{2} \, \ln \, \biggl[ \dfrac{3}{2} \Bigl( 1 - \dfrac{\epsilon_\fin}{3} - \Omega_{\Lambda,0} - \dfrac{\Omega_{\M,0}}{2} \Bigr) + \dfrac{\e^{2 \ln \f/m_\p}}{(1-\Omega_{\Lambda,0})} \dfrac{(\epsilon_\fin - \epsilon_\In)^2}{12 \epsilon_\fin N_{\In 0}^2} \biggr], \\
    \ln \, \dfrac{\m}{H_\fin} & \geq - \ln \dfrac{\f}{m_\p} + \dfrac{1}{2} \, \ln \dfrac{\Lambda_{\mathrm{UV}}^4}{H_\fin^2 m_\p^2} - \dfrac{1}{2} \, \dfrac{c}{\e^{\ln \f/m_\p}}
\end{align*}
If $\smash{f/m_\p < 1}$, and in a regime where the only large parameter is $\smash{\Lambda_{\mathrm{UV}}^4/H_\fin^2 m_\p^2 \gg 1}$, at leading order we can ignore all terms depending on the $\epsilon$-parameters, e-fold numbers and density parameters.
Then, we see that two boundary curves in $(\ln \m/H_\fin, \ln \f/m_\p)$-space intersect at the point such that $\smash{\ln \Lambda_{\mathrm{UV}}^4/H_\fin^2 m_\p^2 \simeq c \, m_\p/\f}$.
So, in practice, this means that we find a minimum mass
\begin{equation} \label{AWGC: mass lower bound}
    \m_{\min} = \dfrac{H_\fin}{c} \, O \Bigl[ \ln \, \Bigl( \dfrac{\Lambda_{\mathrm{UV}}^4}{H_\fin^2 m_\p^2} \Bigr) \Bigr].
\end{equation}
With the observed value $\smash{H_\fin/m_\p = O(10^{-60})}$, we get $\smash{\m_{\min} = (H_\fin/c) \, O [ 276 - 4 \, \ln (m_\p/\Lambda_{\mathrm{UV}})]}$, which manifests the dependencies on $c$ and $\smash{\Lambda_{\mathrm{UV}}}$.
For instance, $\smash{\Lambda_{\mathrm{UV}} = 10^{-10} \, m_\p}$ still implies $\smash{\m_{\min}/H_\fin = O(10^2)/c}$.
This scaling is in place as long as $\smash{\Lambda_{\mathrm{UV}} \gg \e^{-276/4} m_\p = O(10^{-30}) m_\p}$.
A more careful treatment allows one to take care of all the dependencies; most notably, larger absolute values of the AdS shift $\smash{\ab \Lambda \ab}$ raise the minimum allowed mass.

The physical origin of this lower bound resides in the competition between two effects.
On the one hand, generating the observed DE scale from a much larger UV scale requires a strongly suppressed instanton contribution, as per eq.~(\ref{axion mass}); through the AWGC, in eq.~(\ref{axionic WGC}), this favors a small decay constant.
On the other hand, once the axion potential is required to carry a cosmologically relevant energy density, decreasing the decay constant $\f$ forces the mass $\m$ upwards, as apparent from eq.~(\ref{axion quintessence: (m,f)-parameter space bound}).
Indeed, as explained above, the minimum in fig.~\ref{fig.: (m,f)-parameter space bound} occurs where these two tendencies balance.
This fixes $\smash{\m/H_0}$ to scale only logarithmically with the UV-to-IR hierarchy, yielding a mildly UV-sensitive ratio that nevertheless remains parametrically larger than unity.

We stress that eq.~(\ref{AWGC: mass lower bound}) does not apply indiscriminately to every axion in the axiverse.
It applies to an axion that sources the late-time DE evolution considered above, provided that the instanton controlling its potential is subject to the AWGC bound in eq.~(\ref{axionic WGC}).
If its mass were below $\smash{\m_{\min}}$, no value of the decay constant could satisfy both requirements simultaneously: sufficiently large $\f$ would violate the AWGC constraint, whereas sufficiently small $\f$ would violate the dynamical bound required to reproduce the assumed cosmological evolution.
Spectator axions that do not participate in the DE dynamics are not constrained by eq.~(\ref{AWGC: mass lower bound}).

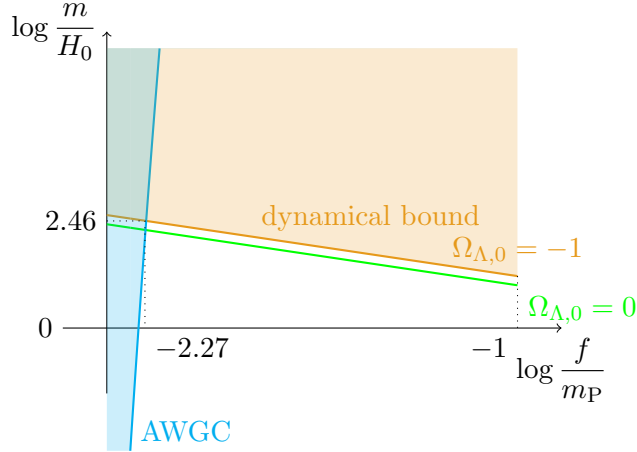
\begin{figure}[h]
    \centering
    
    \begin{tikzpicture}[xscale=3.90,yscale=0.58,every node/.style={font=\normalsize}]

    \draw[->] (-2.55,0) -- (-0.85,0) node[below]{$\log \dfrac{f}{m_\p}$};
    \draw[->] (-2.4,-1.5) -- (-2.4,6.8) node[left]{$\log \dfrac{m}{H_\fin}$};

    \draw[name path = u1, domain=-2.4:-2.32, opacity=0,variable=\x] plot ({\x}, {6.42}) node[above right, black]{};
    \draw[name path = l, domain=-2.4:-2.32, opacity=0,variable=\x] plot ({\x}, {-2.816}) node[above right, black]{};
    \draw[name path = u2, domain=-2.32:-1, opacity=0,variable=\x] plot ({\x}, {6.42}) node[above right, black]{};
    \draw[name path = AdSbound1, domain=-2.4:-2.32, smooth, thick, variable=\x, purple!40!yellow] plot ({\x}, {0.1916-\x});
    \draw[name path = AdSbound2, domain=-2.32:-1, smooth, thick, variable=\x, purple!40!yellow] plot ({\x}, {0.1916-\x}) node[above]{$\Omega_{\Lambda,\fin}=-1$};
    \draw[name path = Mbound1, domain=-2.4:-2.32, smooth, thick, variable=\x, green] plot ({\x}, {-0.019-\x});
    \draw[name path = Mbound2, domain=-2.32:-1, smooth, thick, variable=\x, green] plot ({\x}, {-0.019-\x}) node[below right]{$\Omega_{\Lambda,\fin}=0$};
    \draw[name path = AWGC, domain=-2.32:-2.22, smooth, thick, variable=\x, cyan] plot ({\x}, {40.234 - 0.217147*exp(-2.3026*\x) - \x}) node[above right, black]{};

    \tikzfillbetween[of = u1 and AdSbound1]{purple!40!yellow, opacity=0.2};
    \tikzfillbetween[of = u2 and AdSbound2]{purple!40!yellow, opacity=0.2};
    \tikzfillbetween[of = u1 and l]{cyan, opacity=0.2};
    \tikzfillbetween[of = u2 and AWGC]{cyan, opacity=0.2};

    \node[left] at (-2.55,0){$0$};
    \draw[dotted] (-2.4, 2.457) node[left]{$2.46$} -- (-2.27,2.457) -- (-2.27,0) node[below right] {$-2.27$};
    \draw[dotted] (-1,1.192) -- (-1,0) node[below left] {$-1$};
    \node[purple!40!yellow, above] at (-1.5,2){dynamical bound};
    \node[cyan, above right] at (-2.32,-2.816){AWGC};

    \end{tikzpicture}
    
    \caption{An example of the region allowed in the $(\m,\f)$-parameter space by the combination of eqs.~(\ref{axion quintessence: (m,f)-parameter space bound}, \ref{axionic WGC: (m,f)-parameter space bound}), with $H_\fin=1.4 \cdot 10^{-33} \, \mathrm{eV}$, $\smash{\Omega_{\M,\fin}=1/3}$, and $N_{\In \fin} = 1$, $\epsilon_\fin=2/3$ and $\epsilon_\In=0$ (only the dependence on $H_\fin$ is critical, at small $\f$). The different values $\smash{\Omega_{\Lambda,\fin}=0,-1}$, as expected, show that an AdS shift tightens the bound on the axion mass. For the plot, the values $\Lambda_{\mathrm{UV}} = 10^{-10} \, m_\p$ and $c=1$ have been chosen. The logarithmic plot is in base-$10$ to simplify comparison with data.}
    
    \label{fig.: (m,f)-parameter space bound}
\end{figure}

There exist possible loopholes around our AWGC bound in eq.~(\ref{AWGC: mass lower bound}).
These loopholes have been discussed extensively in the literature, in the context of early cosmic inflation: see e.g. refs.~\cite{delaFuente:2014aca, Rudelius:2015xta, Brown:2015iha, Brown:2015lia, Montero:2015ofa, Hebecker:2015rya} and the recent summary in ref.~\cite{Rudelius:2022gyu}.
For instance, if the instanton providing the dominant contribution to the axion potential is not the AWGC-fulfilling one, practically we would replace $c$ by an effective constant up to $c k$.
However, the understood explicit cases fulfill the AWGC with $k \leq 3$ \cite{Heidenreich:2016aqi, Lee:2019tst}, making this -- as of now -- not a viable way to achieve axion masses $\m = O(H_\fin)$.
If $n$ axions $\phi^a$ are rolling, eq.~(\ref{axiverse AdS-models: physical time-interval lower bound}) leads us to consider $\smash{M_n = \sum_{a=1}^n m_a \, C_{\In \fin}^a}$, where we are ignoring heavier axions as they began to oscillate at earlier times than $t_\In$.
Then, individual axion masses $\m_a = O(H_\fin)$ with subplanckian decay constants $\f_a \ll m_\p$ are compatible with our bound as long as their distribution is dense enough, meaning that there are $n=O(10^2)$ axions around the same mass, if $\smash{C_{\In \fin}^a = O(1)}$.
Moreover, unlike cosmic inflation, which necessitates an effective super-Planckian decay constant in the multi-axion field space, our bound on multi-axion DE models is not a priori in contradiction with the convex-hull formulation of the AWGC \cite{Cheung:2014vva, Rudelius:2015xta, Montero:2015ofa, Brown:2015iha, Brown:2015lia, Heidenreich:2015wga}.
At surface level, the DE scenario is indeed inherently different from that of early inflation, where an effective large decay constant such as $\f_{\mathrm{eff}} \sim \sqrt{n} \f$ would be needed to achieve a long-lived epoch of cosmic acceleration: after all, the decay constant in the axion DE model does not face the same requirements as that of natural inflation, where the need of a long-lived epoch of slow-roll inflation requires $- m_\p^2 V''/V \ab_{\phi/\f=\pi} = m_\p^2/2\f^2 \ll 1$.
If $n$ axions have a similar mass $\m_a \simeq \m$ and the AWGC-fulfilling decay constant is $\f_a \simeq \f$, then the bound in eq.~(\ref{AWGC: mass lower bound}) would simply imply $n \m_{\min}/H_\fin = O(10^2)/c$, allowing for $\m = O(H_\fin)$ if $n=O(10^2)$.
In any case, this loophole does not offer a clear exit either because another problem emerges: the initial conditions to fine-tune increase alongside $n$.
Related tensions implied by the trans-Planckian censorship conjecture \cite{Bedroya:2019snp} for multi-field axion DE models have been discussed in ref.~\cite{Shlivko:2023wkx}.
An alternative loophole relies on engineering a scenario where the axion mass is controlled by poly-instanton terms \cite{Blumenhagen:2008ji}, such as in the model proposed by ref.~\cite{Cicoli:2024yqh}.
In this case, the axion mass acquires a second exponential factor depending on a second decay constant, circumventing the bound in eq. (\ref{AWGC: mass lower bound}).

\subsection{Heuristics} \label{ssec.: heuristics}

The key conclusion from eq.~(\ref{AWGC: mass lower bound}) is that the combination of our bound and AWGC suggests axion masses of order $\m \geq O(10^2) H_\fin/c$.
There is only a mild logarithmic dependence on $\Lambda_{\mathrm{UV}}$.
Although the dependence on $c$ may lower $m_{\min}$ more noticeably, it is still not enough to allow for $m/H_\fin = O(1)$ as long as $c<O(100)$.
Let us interpret this condition heuristically.

A common estimate is that Hubble friction generically freezes the axion as long as
$H \gtrsim O(1) \m$, which motivates the expectation of $\m = O(1) H$ at the onset of oscillations.
A degree of fine tuning of the initial conditions is needed to ensure that the field remains close enough to the hilltop when the DE epoch begins, for $\f/m_\p = O(1)$.
This picture, however, needs not be parametrically unchanged if $\f/m_\p \ll 1$.
Without loss of generality for the current argument, we focus on the case where $\Lambda = 0$.
Based on the field range, we may estimate the scale of the field to be of order $\smash{\ab \phi \ab \sim \f}$; over a Hubble timescale, we can estimate its derivative as $\smash{\ab \dot{\phi} \ab \sim \f H}$.
Away from the minimum, the potential is of order $\smash{V \sim \m^2 \f^2}$, with a derivative $\smash{\ab V' \ab \sim \m^2 \f \, \ab \sin \, (\phi/f) \ab}$.
By writing $\f = \alpha m_\p$, we may thus express the magnitudes involved in the Friedmann equation, eq.~(\ref{Friedmann eq.}), as
\begin{equation} \label{heuristic Friedmann eq.}
    H^2 \sim \alpha^2 \biggl[ H^2 + \m^2 + \dfrac{\sum_i \rho_i}{\alpha^2 m_\p^2} \biggr],
\end{equation}
where at the epoch when the axion potential becomes cosmologically relevant, we assume that the additional fluid energy density is not parametrically dominant.
At the same time, in eq.~(\ref{FRW-KG eq.}), Hubble friction and the force compete as
\begin{equation} \label{heuristic force estimation}
    - 3 H \ab \dot{\phi} \ab + \ab V' \ab \sim \alpha m_\p \, \Bigl[- H^2 + \m^2 \, \Bigl\ab \sin \, \Bigl( \dfrac{\phi}{\f} \Bigr) \Bigr\ab \Bigr].
\end{equation}
If $\alpha = O(1)$, eqs.~(\ref{heuristic Friedmann eq.}, \ref{heuristic force estimation}) suggest that Hubble friction freezes the axion as long as
$H \gtrsim O(1) \m$.
This can in principle accommodate DE today, with $\m = O(1) H_\fin$, provided the hilltop displacement is sufficiently small.\footnote{We may discuss a qualitative comparison with the standard picture for $\smash{f = O(m_\p)}$ and $\smash{\m = O(H_0)}$; see e.g. refs. \cite{Kaloper:2005aj, Dutta:2008qn, Reig:2021ipa}. For instance, ref.~\cite{Kaloper:2005aj} analyzes the hilltop instability in an approximately dS regime, finding a displacement scaling as $\smash{\Delta \phi \simeq A \, \e^{\m t}}$. If the preceding evolution is still partly matter dominated, $H$ is expected to vary more appreciably, especially near the onset and again once the axion kinetic energy becomes relevant. Restricting nevertheless to a quasi-dS epoch, over a Hubble time, with $\smash{H \sim \m \f / m_\p}$, one may therefore estimate $\smash{\Delta \phi \sim A \m / H \sim A m_\p/\f}$. If we further take $\smash{\dot{\phi} \sim f H}$, we find $\smash{A \sim \f^2/m_\p}$, and hence $\smash{\Delta \phi \sim \f}$, in qualitative agreement with an order-$\f$ excursion towards the inflection region.}
However, if $\alpha \ll 1$, the scenario is qualitatively different: once the axion potential becomes cosmologically relevant, the destabilizing force term is parametrically larger for a generic misalignment.
This restricts parametrically more tightly the initial conditions compatible with a frozen field at and past the time at which $H \simeq O(1) m$.
In particular, when $H^2 \sim \alpha^2 \m^2$ -- i.e. when the critical Hubble energy density becomes comparable to the hilltop height, and the kinetic energy is still small, by eq.~(\ref{heuristic Friedmann eq.}) --, the force in eq.~(\ref{heuristic force estimation}) does not parametrically exceed the characteristic Hubble-damping term only as long as $\ab \sin \, (\phi/\f) \ab \lesssim H^2/m^2 \sim \alpha^2$.
This is consistent with the saturation of our bound in eq.~(\ref{axion quintessence: (m,f)-parameter space bound}), by which $\m^2/H_\fin^2 = O(m_\p/\f)^2 = O(1/\alpha^2)$.

\subsection{Larger axion masses and smaller probabilities}

Of course, resorting to an axion living around the hilltop runs into the fine-tuning problem of quantum fluctuations of order $\smash{\delta \phi = H/2\pi}$ during accelerated expansion.
Ways out have been proposed that rely on an axion population whose masses span log-uniformly several orders of magnitude \cite{Kamionkowski:2014zda, Emami:2016mrt}; see also ref.~\cite{Cicoli:2021skd}.

Let us revisit the argument in refs.~\cite{Kamionkowski:2014zda, Emami:2016mrt}, taking $\alpha = O(1/100)$ below.\footnote{We emphasize that the arguments in refs.~\cite{Kamionkowski:2014zda, Emami:2016mrt} were devised for $\alpha = O(1/10)$. As motivated in ssec. \ref{ssec.: heuristics}, our goal is now to show how a change in the order of magnitude of $\alpha$ modifies the parametric probability behavior.}
An optimistic condition for the onset of axion DE is for the initial misalignment to be such that $\smash{\epsilon_V \leq 1}$, solved for $\smash{\ab \pi - \phi_\In / \f \ab \leq 2 \sqrt{2} \alpha}$.
Assuming a uniform distribution of initial conditions, the probability to fall in this interval is $p_1(\alpha) = 2 \sqrt{2} \alpha/\pi$.
For concreteness, let us consider a theory with an axion population $\phi^a$ whose masses $m_a$ span several orders of magnitude, with a label $a$ such that $m_a > m_{a+1}$ for the $a$-th and $(a+1)$-th axions.
Then, for a given label $a_\fin$, the probability of having the $a_\fin-1$ heaviest axions too far from the hilltop and at the same time finding the $a_\fin$-th axion close enough is $\smash{P_1(\alpha,a_\fin) = p_1(\alpha) [1 - p_1(\alpha)]^{a_\fin-1}}$, scaling like
\begin{equation}
    P_1(\alpha,a_\fin) = O[\alpha (1-\alpha)^{a_\fin-1}].
\end{equation}
As an example, we can consider non-degenerate masses of the form in eq.~(\ref{axion mass}), for $k=1$ and with $S_a = 10 a$ and $\smash{\Lambda_{\mathrm{UV}} = 10^{-10} \, m_\p}$.
In this case, one finds $\smash{m_{19} = O(H_\fin)}$ (meaning that the instanton action $S_a$ that generates the axion mass has $a=19$).
The probability of this scenario to induce DE today is $P_1=O(10^{-2})$, which is remarkably higher than one may naively have expected.

However, with a largely subplanckian decay constant, the heuristic argument above points to a different scaling.
As we argued, we require $\smash{\ab \sin \, (\phi_\In/\f) \ab \lesssim \alpha^2}$, i.e. $\smash{\ab \pi - \phi_\In / \f \ab \lesssim \alpha^2}$, as a pre-condition for the onset of a DE epoch.
This implies a quadratic probability $p_2(\alpha) = O(\alpha^2)$, rather than a linear one, and
\begin{equation} \label{probability}
    P_2(\alpha,a_\fin) = O[\alpha^2 (1-\alpha^2)^{a_\fin-1}].
\end{equation}
In the same setup as above, we find $\smash{m_{18} = O(10^{3}) H_\fin}$, leading to a probability $P_2 = O(10^{-4})$ of finding DE today within this line of reasoning.

\section{Discussion} \label{sec.: discussion}

In this paper, we study analytic solutions to models of the current universe where DE is represented by a periodic field rolling down its potential, gravitationally coupled to additional perfect cosmological fluids such as baryonic and dark matter.
A simple example of this scenario is a single axion-like particle with a cosine potential, but our results can apply equally well to potentials with arbitrary powers of the periodic functions and with additional contributions from a cosmological constant of either sign.

Focusing on transient solutions, we 
are able to derive a relationship between the field couplings and the boundary conditions, including the values of the Hubble parameter and the density parameters associated with the various contributions to the total energy density, as the universe relaxes its acceleration rate toward its present value.
Notably, our analysis does not rely on specific initial conditions, aside from the mild assumption that the current acceleration rate is lower than in the past.
To illustrate the utility of our results, we construct exclusion plots in the parameter space defined by the DE equation-of-state parameter and the matter density parameter, identifying regions consistent with present-day observations within both axion and AdS-shifted axion DE models.

Our analytic bound, when combined with quantum gravity expectations, suggests an axion mass that is somewhat in tension with the value suggested by observations.
In particular, along with sub-Planckian axion decay constants implied by the AWGC, the combination of our bound and of the AWGC points to an axion mass scale of order $m = O(10^2) H_0/c$, with $c=O(1)$, which is roughly two orders of magnitude larger than the estimate $m = O(1) H_0$ for the quintessence field, appearing in several phenomenological axion-DE fits \cite{Luu:2025fgw, Urena-Lopez:2025rad}.
This highlights a challenge in embedding phenomenological axion DE models into UV-complete theories.
It should be stressed, however, that the observational preference for $m$ is sensitive to the treatment of the initial conditions and priors.
In particular, under an alternative sampling strategy, ref. \cite{Lin:2025gne} finds a possible preference for a relatively large axion mass, of order $m = O(10) H_0$, which reduces the gap with the outcome of our discussion.
We have also discussed the tuning of initial conditions in this region of parameter space.
Heuristically, the small-$\f$/large-$\m$ regime suggests that the hilltop window required for late-time acceleration is more finely tuned than in the case of only mildly sub-Planckian $\f$: the tuning does not merely worsen because $\f$ is smaller, but also through an additional parametric suppression.

The dynamical-system framework we have developed is highly adaptable and it can accommodate additional potentials and couplings.
A natural extension involves promoting the saxion that is responsible for setting the axion decay constant to a fully dynamical field \cite{Cicoli:2020cfj}, as well as incorporating negative spatial curvature \cite{Andriot:2024jsh}.
Furthermore, the phantom crossing in the DE sector indicated by the DESI data can be consistently explained without violating the null energy condition via a coupling of DE to dark matter \cite{Das:2005yj}.
These kinds of models are in fact motivated from a UV perspective and show agreement with data \cite{Agrawal:2019dlm, Bedroya:2025fwh}.
In light of the machinery we have worked out, an analytic treatment of transient solutions in these more sophisticated scenarios becomes feasible, including multi-field settings
\cite{Shiu:2024sbe, Shiu:2025ycw}.

Ultra-light axions with periodic potentials of the form in eq.~(\ref{generic V}) are also leading candidates to resolve the $H_\fin$-tension in EDE models, with the best fit to data selecting $p=3$ and $\f/m_\p = O(1/10)$ \cite{Poulin:2018cxd, McDonough:2021pdg} (see however refs.~\cite{Rudelius:2022gyu, Cicoli:2023qri} for an illustration of the tensions in finding such values in UV-complete theories).
To a first approximation, we may fix the boundary conditions at a final time $\smash{t_\fin = t_{\rec}}$ as $\smash{(\Omega_{\M, \rec}, \Omega_{\rad, \rec}) = (0.76,0.24)}$, at redshift $\smash{z_{\rec} = 1.1 \cdot 10^3}$, with effectively $\smash{\Omega_{\phi, \rec} \simeq 0}$, which corresponds to $\smash{\epsilon_{\rec} = 1.62}$.
Although this means that the bound derived above cannot be applied directly, the underlying framework and methods that we have discussed remains fundamentally applicable.
In particular, one can formulate the axion-matter-radiation system in the same autonomous variables, derive differential inequalities controlling the growth of the axion kinetic energy and the dilution of the background fluids, and then translate them into restrictions on the duration and magnitude of the EDE epoch. The main modification consists in adapting the step that presently assumes $\smash{\epsilon \leq 3/2}$ to the radiation-matter background around recombination.
This might provide analytic constraints on the axion parameters and initial displacement that complement the existing numerical exploration of EDE models and the $H_0$ tension \cite{DiValentino:2021izs}.

\acknowledgments
We would like to thank Mario Reig and Philip Sørensen for useful conversations.
GS is supported in part by the DOE grant DE-SC0017647.
FT is funded by the European Union -- grant NextGenerationEU/PNRR mission 4.1: CUP C93C24004950006.
HVT is supported in part by NSF grant DMS-2348305.

\appendix

\section{Alternative coordinates for the dynamical system} \label{app: alternative coordinates for the dynamical system}

In complex coordinates, exponential and trigonometric functions are indistinguishable.
This reveals a mathematical connection between periodic and exponential runaway potentials which we elucidate below.
Different coordinate choices present in earlier literature are also discussed.

\subsection{Complex coordinates}

Let us consider the complex potential
\begin{equation} \label{complex V}
    V_{\mathbb{C}} = \Lambda_\fin - \dfrac{\Lambda_+}{2} \, \e^{\I \kappa_d \gamma_+ \phi} - \dfrac{\Lambda_-}{2} \, \e^{-\I \kappa_d \gamma_- \phi}.
\end{equation}
A classical physical system requires a real classical action, but for generality let $\gamma_\pm$ and $\Lambda_\fin, \Lambda_\pm$ be unrelated to each other.
One may uncover structures in the dynamical system for an over-rich problem, i.e. a complex potential $V_{\mathbb{C}}$, and then specify which real potential one is considering.
For instance:
\begin{itemize}
    \item for $\Lambda_\fin = \Lambda_\pm = \m^2 \f^2$ and $\gamma_\pm = \gamma$, the potential reduces to eq.~(\ref{V});
    \item for $\Lambda_\fin = 0$, $\Lambda_\pm = - \Lambda$ and $\gamma_\pm = \pm \I \gamma$, the potential reduces to a simple exponential potential $V_{\exp} = \Lambda \, \e^{- \kappa_d \gamma \phi}$.
\end{itemize}
Below, we define a dynamical system associated to the complex potential in eq.~(\ref{complex V}).
Later on, we will study solutions for parameter choices that make the potential real.

Let
\begin{subequations}
\begin{align}
    x & = \dfrac{\kappa_d}{\sqrt{d-1} \sqrt{d-2}} \, \dfrac{\dot{\phi}}{H}, \label{complex DS - x} \\
    y_\pm & = \pm \I \dfrac{\kappa_d \sqrt{\Lambda_\pm}}{\sqrt{d-1} \sqrt{d-2}} \, \dfrac{1}{H} \, \e^{\pm \frac{\I}{2} \kappa_d \gamma_\pm \phi}, \label{complex DS - y} \\
    z & = \dfrac{\kappa_d \sqrt{2}}{\sqrt{d-1} \sqrt{d-2}} \, \dfrac{\sqrt{\rho}}{H}, \label{complex DS - z}
\end{align}
\end{subequations}
noticing that the functions $x$ and $z$ are real, with instead complex-valued functions $y_\pm$, and
\begin{subequations}
\begin{align}
    h & = (d-1) H, \label{complex DS - xf} \\[0.85ex]
    c_\pm & = \dfrac{1}{2} \dfrac{\sqrt{d-2}}{\sqrt{d-1}} \, \gamma_\pm, \label{complex DS - c} \\
    l_\fin & = \dfrac{\sqrt{d-1}}{\sqrt{d-2}} \sqrt{2 \smash{\kappa_d^2} \Lambda_\fin}. \label{complex DS - l}
\end{align}
\end{subequations}
Then, the cosmological problem in eqs.~(\ref{FRW-KG eq.}, \ref{continuity eq.}, \ref{Friedmann eq.}) is equivalent to the equations
\begin{subequations}
\begin{align}
    \dot{x} & = \biggl[ -x - \I \, \bigl[ c_+ (y_+)^2 - c_- (y_-)^2 \bigr] + x \, \Bigl[ (x)^2 + \dfrac{1+w}{2} \, (z)^2 \Bigr] \biggr] \, h, \label{complex DS - x-equation} \\
    \dot{y}_\pm & = \biggl[ (x)^2 + \dfrac{1+w}{2} \, (z)^2 \pm \I c_\pm x \biggr] \, y_\pm h, \label{complex DS - y-equation} \\[0.45ex]
    \dot{z} & = \biggl[ - \dfrac{1+w}{2} + (x)^2 + \dfrac{1+w}{2} \, (z)^2 \biggr] \, z h, \label{complex DS - z-equation} \\
    \dot{h} & = - \biggl[ (x)^2 + \dfrac{1+w}{2} \, (z)^2 \biggr] h^2, \label{complex DS - h-equation}
\end{align}
\end{subequations}
jointly with the condition
\begin{equation}
    (x)^2 + (y_+)^2 + (y_-)^2 + (z)^2 = 1 - \dfrac{l_\fin^2}{h^2}. \label{complex DS - sphere condition}
\end{equation}

A way to check the correctness of eqs.~(\ref{complex DS - x-equation}, \ref{complex DS - y-equation}, \ref{complex DS - z-equation}, \ref{complex DS - h-equation}, \ref{complex DS - sphere condition}) is to notice that the potential in eq.~(\ref{complex V}) reduces to a simple exponential potential, as already mentioned, for $\Lambda_\fin \mapsto 0$, $\Lambda_\pm \mapsto - \Lambda$ and $\gamma_\pm \mapsto \pm \I \gamma$.
Through the identifications $l_\fin \mapsto 0$, $c_\pm \mapsto \pm \I c$ and $y_\pm \mapsto \mp y/\sqrt{2}$, one finds \cite{Copeland:1997et}
\begin{subequations}
\begin{align}
    \dot{x} & = \biggl[ -x + c \, (y)^2 + x \, \Bigl[ (x)^2 + \dfrac{1+w}{2} \, (z)^2 \Bigr] \biggr] \, h, \label{exponential-potential x-equation - complex notation} \\
    \dot{y} & = \biggl[ (x)^2 + \dfrac{1+w}{2} \, (z)^2 - c x \biggr] \, y h, \label{exponential-potential y-equation - complex notation} \\[0.45ex]
    \dot{z} & = \biggl[ - \dfrac{1+w}{2} + (x)^2 + \dfrac{1+w}{2} \, (z)^2 \biggr] \, z h, \label{exponential-potential z-equation - complex notation} \\
    \dot{h} & = - \biggl[ (x)^2 + \dfrac{1+w}{2} \, (z)^2 \biggr] h^2, \label{exponential-potential h-equation - complex notation}
\end{align}
\end{subequations}
jointly with the condition
\begin{equation}
    (x)^2 + (y)^2 + (z)^2 = 1. \label{exponential-potential sphere condition - complex notation}
\end{equation}
For multi-field multi-exponential potentials, if $z=0$, universal analytic convergence results were discussed in ref.~\cite{Shiu:2023fhb}.

For the potential in eq.~(\ref{V}), one gets the dynamical system
\begin{subequations}
\begin{align}
    \dot{x} & = \biggl[ -x - \I c \, \bigl[(y_+)^2 - (y_+^*)^2 \bigr] + x \, \Bigl[ (x)^2 + \dfrac{1+w}{2} \, (z)^2 \Bigr] \biggr] \, h, \label{x-equation - complex notation} \\
    \dot{y}_+ & = \biggl[ (x)^2 + \dfrac{1+w}{2} \, (z)^2 + \I c x \biggr] \, y_+ h, \label{y-equation - complex notation} \\[0.45ex]
    \dot{z} & = \biggl[ - \dfrac{1+w}{2} + (x)^2 + \dfrac{1+w}{2} \, (z)^2 \biggr] \, z h, \label{z-equation - complex notation} \\
    \dot{h} & = - \biggl[ (x)^2 + \dfrac{1+w}{2} \, (z)^2 \biggr] h^2, \label{h-equation - complex notation}
\end{align}
\end{subequations}
subject to the two constraints
\begin{subequations}
\begin{align}
    (x)^2 + (y_+)^2 + (y_+^*)^2 + (z)^2 & = 1 - \dfrac{l^2}{h^2}, \label{sphere condition - complex notation} \\
    y_+ y_+^* & = \dfrac{1}{2} \dfrac{l^2}{h^2}, \label{reality-condition - complex notation}
\end{align}
\end{subequations}
where one should notice that $\smash{y_- = y_+^*}$.
Such two constraints can be combined into
\begin{equation}
    (x)^2 + (y_+ + y_+^*)^2 + (z)^2 = 1. \label{combined constraint - complex notation}
\end{equation}

Let $\smash{\yr = (y_+ + y_+^*)/2}$ and $\smash{\yi = -\I (y_+ - y_+^*)/2}$, i.e.
\begin{equation}
    y_+ = \yr + \I \yi.
\end{equation}
Then one gets immediately the dynamical system in eqs.~(\ref{x-equation}, \ref{yr-equation}, \ref{yi-equation}, \ref{z-equation}, \ref{h-equation}, \ref{sphere condition}).

\subsection{Comparison with the literature}
As summarized by recent refs.~\cite{DESI:2025fii, Urena-Lopez:2025rad}, in terms of different variables, studies of the same problem have been performed before in terms of an autonomous dynamical system.
Taking as a reference refs.~\cite{Cedeno:2017sou, LinaresCedeno:2020dte}, one may define polar coordinates through the identifications\footnote{It should be noted that refs.~\cite{Cedeno:2017sou, LinaresCedeno:2020dte, DESI:2025fii} define the potential as $\smash{V(\varphi) = m^2 f_\varphi^2 \, \bigl[ 1 + \cos \, (\varphi/f_\varphi) \bigr]}$. Hence, the fields $\phi$ and $\varphi$ are related as $\smash{\phi = \varphi + f_\varphi \pi}$. Also, for convenience we defined the scalar field density parameter as $\Omega = \sigma^2$.}
\begin{subequations}
    \begin{align}
        x & = \sigma \, \cos \, \dfrac{\theta}{2}, \\
        \yr & = - \dfrac{\sigma}{2} \, \sin \dfrac{\theta}{2},
    \end{align}
\end{subequations}
as well as
\begin{equation}
    \yi = - \dfrac{c}{2} \psi.
\end{equation}
By writing the inverse relations
\begin{subequations}
    \begin{align}
        \sigma & = \sqrt{(x)^2 + (2 \yr)^2}, \\
        \theta & = 2 \arctan \, \biggl( - \dfrac{1}{2} \dfrac{x}{\yr} \biggr),
    \end{align}
\end{subequations}
it is obvious that one can rearrange eqs.~(\ref{x-equation}, \ref{yr-equation}, \ref{yi-equation}, \ref{z-equation}, \ref{h-equation}, \ref{sphere condition}) as\footnote{It might be useful to note the identity
\begin{align*}
    w_{\mathrm{tot}} = - 1 + \dfrac{2 \epsilon}{3} = \dfrac{T - V + w \rho}{T + V + \rho} = w - (w + \cos \theta) \sigma^2.
\end{align*}}\begin{subequations}
    \begin{align}
        \dot{\sigma} & = \dfrac{1}{2} \, (w + \cos \theta) (1 - \sigma^2) \, \sigma h, \\
        \dot{\theta} & = (- \sin \theta + 2 \psi) h, \\
        \dot{\psi} & = \dfrac{1}{2} \, \bigl[ [(1+w) - (w + \cos \theta) \sigma^2] + \sigma^2 \sin \theta \bigr] h, \\
        \dot{z} & = - \dfrac{1}{2} (w + \cos \theta) [1 - (z)^2] z h, \\
        \dot{h} & = - \dfrac{1}{2} \, [(1+w) - (w + \cos \theta) \sigma^2] h^2,
    \end{align}
\end{subequations}
with the constraint
\begin{equation}
    \sigma^2 + (z)^2 = 1.
\end{equation}

\subsection{Comparison with slow-roll results}
We can compare our results with the estimates one would get under the slow-roll approximation.
Let us neglect the additional fluid, and, for definiteness, let us consider solutions where $\phi/\f \in \; [0,\pi]$ and $\dot{\phi}<0$.
In the slow-roll regime, one can approximate
\begin{subequations}
    \begin{align}
        \kappa_d \dot{\phi} & = - \sqrt{d-2} \, \sqrt{\epsilon_V} H, \\
        H^2 & = \dfrac{2 \kappa_d^2 V}{(d-1)(d-2)},
    \end{align}
\end{subequations}
where
\begin{equation}
    \epsilon_V = \dfrac{d-2}{4 \kappa_d^2 \f^2} \, \cot^2 \Bigl(\dfrac{\phi}{2 \f}\Bigr),
\end{equation}
As long as $\smash{\ab \pi - \phi/\f \ab \leq \ab \pi - 2 \, \cot^{-1} \, (2 \kappa_d \f/\sqrt{d-2}) \ab}$, we have $\epsilon_V \leq 1$; in particular, for small $\kappa_d \f \ll 1$, we may approximate $\smash{\ab \pi - \phi/\f \ab \lesssim 4 \kappa_d \f/\sqrt{d-2}}$.
A simple integration gives
\begin{equation} \label{slow-roll physical time-interval lower bound}
    \Delta t_{\In \fin} = \dfrac{\kappa_d \f}{\m} \dfrac{\sqrt{d-1}}{\sqrt{d-2}} \, \ln \, \dfrac{1 + \sin \, \dfrac{\phi_\In}{2 \f}}{1 - \sin \, \dfrac{\phi_\In}{2 \f}} \, \dfrac{1 - \sin \, \dfrac{\phi_\fin}{2 \f}}{1 + \sin \, \dfrac{\phi_\fin}{2 \f}}.
\end{equation}
Although eq.~(\ref{slow-roll physical time-interval lower bound}) expresses neatly the time interval as a function of the initial and final misalignment angles, it is only accurate for small such angles and neglecting additional fluids.
To gain more intuition, let us expand for small deviations from the hilltop, and express the result in terms of the $\smash{\epsilon_V}$-parameter via $\smash{\phi = 2 \f \, \mathrm{arccot} \, \bigl[ 2 \kappa_d \f \, \sqrt{\epsilon_V}/\sqrt{d-2} \bigr]}$.
We obtain
\begin{equation} \label{approximate slow-roll physical time-interval}
    \Delta t_{\In \fin} = \dfrac{\kappa_d \f}{\m} \dfrac{\sqrt{d-1}}{\sqrt{d-2}} \, \biggl[ \ln \dfrac{\epsilon_{V, \fin}}{\epsilon_{V, \In}} + \sum_{n=1}^\infty \dfrac{(\kappa_d \f)^{2n}}{c_n} \dfrac{\epsilon_{V, \fin}^n - \epsilon_{V, \In}^n}{(d-2)^n} \biggr],
\end{equation}
where $c_1 = -1/2$, $c_2 = 1/3$, $c_3 = -3/20$ and $c_4 = 2/35$.

Let us now rewrite eq.~(\ref{refined physical time-interval lower bound}) in a way that is easy to compare with the slow-roll result.
In general, we may express eq.~(\ref{refined physical time-interval lower bound}) as\footnote{In the absence of additional fluids, one may always write
\begin{align*}
    \cos^2 \Bigl[\dfrac{\phi}{2 \f}\Bigr] = 1 - \dfrac{d-2}{4 \kappa_d^2 \f^2 \m^2} \, (d-1-\epsilon) H^2.
\end{align*}
Then, assuming the slow-roll conditions to hold, one can approximate $\smash{\epsilon \simeq \epsilon_V}$ and $\smash{H^2 \simeq 4 \kappa_d^2 \f^2 \m^2 \sin^2 (\phi/2\f)/[(d-1)(d-2)]}$.}
\begin{equation}
    \Delta t_{\In \fin} \gtrsim \dfrac{\dfrac{1}{2 \m} \dfrac{1}{\sqrt{d-1}} \biggl[\sqrt{\epsilon_{V,\fin}} - \dfrac{\epsilon_{V,\In}}{\sqrt{\epsilon_{V,\fin}}}\biggr]}{\sqrt{1 - \dfrac{d-2}{d-1} \dfrac{d-1 - \epsilon_{V,\fin}}{d-2 + 4 \kappa_d^2 \f^2 \epsilon_{V,\fin}}}}.
\end{equation}
Again, expanding for small $\epsilon_V$-parameters, we eventually find the next-to-leading-order expression in $\smash{\epsilon_{V,\In}/\epsilon_{V,\fin}}$
\begin{equation} \label{approximate slow-roll physical time-interval lower bound}
    \Delta t_{\In \fin} \gtrsim \dfrac{1}{2 \m} \dfrac{\sqrt{d-2}}{\sqrt{d-2 + 4 (d-1) \kappa_d^2 \f^2}} \, \biggl( 1 - \dfrac{\epsilon_{V,\In}}{\epsilon_{V,\fin}} \biggr).
\end{equation}
A comparison is displayed in fig.~\ref{fig.: slow-roll comparison}.
The inherently slow-roll result in eq.~(\ref{approximate slow-roll physical time-interval}) belongs to the interval singled out by the slow-roll approximation of our result, in eq.~(\ref{approximate slow-roll physical time-interval lower bound}).
Yet, the inherently slow-roll result only works neglecting additional fluids, and only for small $\epsilon_V$s, whereas our result is always valid and hence provides a more advanced tool to study longer-lived epochs of cosmic evolution.

\begin{figure}[h]
    \centering
    
    \begin{tikzpicture}[scale=4.5,every node/.style={font=\normalsize}]

    \draw[->] (0.90,0) -- (2.1,0) node[right]{$\dfrac{\epsilon_{V,\fin}}{\epsilon_{V,\In}}$};
    \draw[->] (1,-0.1) -- (1,0.7) node[left]{$m \Delta t_{\In \fin}$};
    
    \draw[name path = a, domain=1:2, smooth, thick, variable=\x, green, opacity=0] plot ({\x}, {0.65});
    \draw[name path = b, domain=1:2, smooth, thick, variable=\x, green] plot ({\x}, {0.287*(1-1/\x)}) node[above right, black]{};
    \tikzfillbetween[of = a and b]{green, opacity=0.1};
    
    \draw[domain=1:2, smooth, thick, variable=\x, orange] plot ({\x}, {0.707*ln(\x)}) node[above right, black]{};

    \node[above left] at (1,0) {$0$};
    \node[below right] at (1,0) {$1$};

    \end{tikzpicture}
    
    \caption{For a fixed value of $\kappa_d f$, the orange curve represents the slow-roll calculation of $m \Delta t_{\In \fin}$, in eq.~(\ref{approximate slow-roll physical time-interval}), while the green area represents the region allowed by our bound as approximated in the slow-roll regime, in eq.~(\ref{approximate slow-roll physical time-interval lower bound}).}
    
    \label{fig.: slow-roll comparison}
\end{figure}
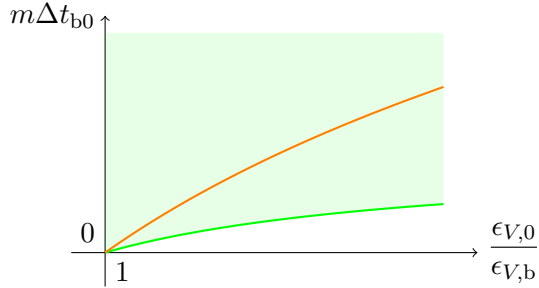

\section{Refined bounds} \label{app.: refined bounds}

In the derivation of the universal bound of eq.~(\ref{refined physical time-interval lower bound}), at the cost of getting cluttered formulae, we may actually get a sharper result, as anticipated in the main text.

As a preliminary result, let us observe the following trivial fact.

\begin{lemma} \label{lemma: h lower bound}
    At all times $t \in [t_\In, t_\fin]$, we have
    \begin{equation} \label{eq.: h lower bound}
        h(t) \geq \dfrac{1}{\xi_\fin (t - t_\In) + \dfrac{1}{h_\In}}.
    \end{equation}
\end{lemma}

\begin{proof}
    By assumption, we have $\xi \leq \xi_\fin$ at all times $t \in [t_\In, t_\fin]$.
    Hence, eq.~(\ref{h-equation}) allows one to write $\dot{h} = - \xi h^2 \geq - \xi_\fin h^2$, whose integration immediately gives eq.~(\ref{eq.: h lower bound}).
\end{proof}

As a first improvement, we note that, in constraining $\ab x \ab$ as we do in the proof of lemma \ref{lemma: abx upper bound}, we may keep the linear term in $\ab x \ab$.

\begin{corollary} \label{corollary: improved abx upper bound}
    One has
    \begin{equation} \label{eq.: improved abx upper bound}
        \ab x(t) \ab \leq \dfrac{1}{\bigl[ 1 + \xi_\fin h_\In (t-t_\In) \bigr]^{\frac{1}{\xi_\fin}-1}} \, \biggl[ \ab x_\In \ab - \dfrac{\sqrt{2} \, c l}{h_\In} \biggr] + \sqrt{2} \, c l \, \Biggl[ \xi_\fin (t-t_\In) + \dfrac{1}{h_\In} \Biggr].
    \end{equation}
    This reduces to the statement of lemma \ref{lemma: abx upper bound}, i.e. eq.~(\ref{eq.: abx upper bound}), by formally fixing $\xi_\fin \mapsto 1$ in the right hand-side.
\end{corollary}

\begin{proof}
    In view of eqs.~(\ref{x-equation}, \ref{eq.: h lower bound}), we may write the chain of differential inequalities
    \begin{align*}
        \dfrac{\de}{\de t} \, \ab x \ab = \bigl[ - \ab x \ab (1- \xi) + 4 \mathrm{sgn}(x) \, c \, \yr \yi \bigr] \, h & \leq \bigl[ - \ab x \ab (1- \xi_\fin) + 4 \mathrm{sgn}(x) \, c \, \yr \yi \bigr] \, h \\
        & \leq - \dfrac{(1 - \xi_\fin) \, \ab x \ab}{\xi_\fin (t - t_\In) + \dfrac{1}{h_\In}} + \sqrt{2} \, c l,
    \end{align*}
    where we also took advantage of eqs.~(\ref{sphere condition}, \ref{(yr,yi)-bounds}) to bound the term depending on $c$.
    An integration immediately gives eq.~(\ref{eq.: improved abx upper bound}).
\end{proof}

As a second improvement, in constraining $z$, we can constrain the functional form of $z$, beyond just observing it decreases as done in lemma \ref{lemma: time-interval lower bound}.
As in the main text, let $\xi_\fin \leq (1+w)/2$.

\begin{corollary} \label{corollary: improved z upper bound}
    One has
    \begin{equation} \label{eq.: improved z upper bound}
        z(t) \leq \dfrac{z_\In}{\bigl[ 1 + \xi_\fin h_\In (t-t_\In) \bigr]^{\frac{1+w}{2 \xi_\fin} - 1}}.
    \end{equation}
\end{corollary}

\begin{proof}
    In view of eqs.~(\ref{z-equation}, \ref{eq.: h lower bound}), because $\xi \leq \xi_\fin$, one may write the differential inequality
    \begin{align*}
        \dot{z} = \Bigl[ - \dfrac{1+w}{2} + \xi \Bigr] z h \leq \Bigl[ - \dfrac{1+w}{2} + \xi_\fin \Bigr] z h \leq \dfrac{- \dfrac{1+w}{2} + \xi_\fin}{\xi_\fin (t - t_\In) + \dfrac{1}{h_\In}} \, z.
    \end{align*}
    An integration immediately gives eq.~(\ref{eq.: improved z upper bound}).
\end{proof}

Finally, combining corollaries \ref{corollary: improved abx upper bound} and \ref{corollary: improved z upper bound} allows us to reach a stronger conclusion than lemma \ref{lemma: time-interval lower bound}.

\begin{lemma} \label{lemma: improved time-interval lower bound}
    The duration of the time interval $[t_\In, t_\fin]$ is related to the variation of the $\xi$-parameter as
    \begin{equation} \label{eq.: improved time-interval lower bound}
        \begin{split}
            \xi(t) - \xi_\In \leq 2 \sqrt{\xi(t)} \, \Biggl[ \sqrt{2} \, c l \, \xi_\fin (t-t_\In) + \dfrac{1 - \bigl[ 1 + \xi_\fin h_\In (t-t_\In) \bigr]^{\frac{1}{\xi_\fin}-1}}{\bigl[ 1 + \xi_\fin h_\In (t-t_\In) \bigr]^{\frac{1}{\xi_\fin}-1}} \, \biggl[ \ab x_\In \ab - \dfrac{\sqrt{2} \, c l}{h_\In} \biggr] \Biggr] \\
            + \dfrac{1+w}{2} \, (z_\In)^2 \Biggl[ \dfrac{1}{\bigl[ 1 + \xi_\fin h_\In (t-t_\In) \bigr]^{\frac{1+w}{\xi_\fin} - 2}} - 1 \Biggr] &.
        \end{split}
    \end{equation}
    This reduces to the statement of lemma \ref{lemma: time-interval lower bound} by formally fixing $(z_\In, \xi_\fin) \mapsto (0,1)$ in the right hand-side.
\end{lemma}

\begin{proof}
    Let us split the proof into two main steps.
    \begin{enumerate}[label=\roman*.]
        \item Because $(x)^2 - (x_\In)^2 = \xi - \xi_\In - [(1+w)/2] \, [(z)^2 - (z_\In)^2]$, in view of eq.~(\ref{eq.: improved z upper bound}), we may write
        \begin{align*}
            [x(t)]^2 - (x_\In)^2 \geq \xi(t) - \xi_\In - \dfrac{1+w}{2} \, (z_\In)^2 \Biggl[ \dfrac{1}{\bigl[ 1 + \xi_\fin h_\In (t-t_\In) \bigr]^{\frac{1+w}{\xi_\fin} - 2}} - 1 \Biggr].
        \end{align*}
        \item Because $(x)^2 - (x_\In)^2 = [\ab x \ab + \ab x_\In \ab] [\ab x \ab - \ab x_\In \ab]$, in view of eq.~(\ref{eq.: improved abx upper bound}), we can write
        \begin{align*}
            [x(t)]^2 - (x_\In)^2 \leq 2 \sqrt{\xi(t)} \, \Biggl[ \sqrt{2} \, c l \, \xi_\fin (t-t_\In) + \dfrac{1 - \bigl[ 1 + \xi_\fin h_\In (t-t_\In) \bigr]^{\frac{1}{\xi_\fin}-1}}{\bigl[ 1 + \xi_\fin h_\In (t-t_\In) \bigr]^{\frac{1}{\xi_\fin}-1}} \, \biggl[ \ab x_\In \ab - \dfrac{\sqrt{2} \, c l}{h_\In} \biggr] \Biggr],
        \end{align*}
        where we also took advantage of the inequality $\smash{(x)^2 \leq \xi - [(1+w)/2] \, (z)^2 \leq \xi}$.
    \end{enumerate}
    By combining the two bounds above, eq.~(\ref{eq.: improved time-interval lower bound}) follows immediately.
\end{proof}

\section{Generalized periodic potentials} \label{app.: generalized bounds}

In order to study the generalized model of eq.~(\ref{generalized V}), one may exploit the results for the elementary potential of eq.~(\ref{V}), adapting them accordingly.
One can write
\begin{subequations}
\begin{align}
    V_{p} & = \dfrac{V^p}{(\m^2 \f^2)^{p-1}}, \\
    V'_{p} & = p \, \dfrac{V^{p-1}}{(\m^2 \f^2)^{p-1}} \, V', \\
    V''_{p} & = p \, \dfrac{V^{p-1}}{(\m^2 \f^2)^{p-1}} \, \biggl[ V'' + (p-1) \, \dfrac{(V')^2}{V} \biggr].
\end{align}
\end{subequations}
One may define an autonomous dynamical system through the same variable redefinitions as in eqs.~(\ref{x}, \ref{yr}, \ref{yi}, \ref{z}, \ref{h}) and eqs.~(\ref{c}, \ref{l}).
Let us also note that
\begin{align*}
    \dfrac{\kappa_d^2 V_{p}}{(d-1)(d-2) H^2} & = 2^{2p-1} (\yr)^{2p} \, \Bigl( \dfrac{h^2}{l^2} \Bigr)^{p-1}, \\
    \dfrac{\f \kappa_d^2 V'_{p}}{(d-1)(d-2) H^2} & = - 2^p p \, \yr^{2p-1} \yi \, \Bigl( \dfrac{h^2}{l^2} \Bigr)^{p-1}, \\
    \dfrac{\f^2 \kappa_d^2 V''_{p}}{(d-1)(d-2) H^2} & = 2^p p \, \biggl[ 2^{p} (p-1) \, (\yi)^2 - \dfrac{1}{2} \, [(\yr)^2 - (\yi)^2] \, \biggr] \, (\yr)^{2(p-1)} \Bigl( \dfrac{h^2}{l^2} \Bigr)^{p-1}.
\end{align*}

It might as well be useful to take into account several barotropic fluids, with energy densities $\rho_\alpha$ and parameters $w_\alpha$, each one with a corresponding variable $z_\alpha$ defined as in eq.~(\ref{z}).

One can verify that the dynamical system takes the form
\begin{subequations}
\begin{align}
    \dot{x} & = \Biggl[ -x + 2^{p+1} p c \, (\yr)^{2p-1} \yi \, \Bigl( \dfrac{h^2}{l^2} \Bigr)^{p-1} + x \, \biggl[ (x)^2 + \sum_{\beta} \dfrac{1+w_\beta}{2} \, (z_\beta)^2 \biggr] \Biggr] \, h, \label{generalized x-equation} \\
    \yrd & = \Biggl[ \biggl[ (x)^2 + \sum_{\beta} \dfrac{1+w_\beta}{2} \, (z_\beta)^2 \biggr] \yr - c x \yi \Biggr] \, h, \label{generalized yr-equation} \\[0.45ex]
    \yid & = \Biggl[ \biggl[ (x)^2 + \sum_{\beta} \dfrac{1+w_\beta}{2} \, (z_\beta)^2 \biggr] \yi + c x \yr \Biggr] \, h, \label{generalized yi-equation} \\[0.45ex]
    \dot{z}_\alpha & = \Biggl[ - \dfrac{1+w_\alpha}{2} + \biggl[ (x)^2 + \sum_{\beta} \dfrac{1+w_\beta}{2} \, (z_\beta)^2 \biggr] \Biggr] \, z_\alpha h, \label{generalized z-equation} \\
    \dot{h} & = - \biggl[ (x)^2 + \sum_{\beta} \dfrac{1+w_\beta}{2} \, (z_\beta)^2 \biggr] h^2, \label{generalized h-equation}
\end{align}
\end{subequations}
subject to the constraint
\begin{equation}
    (x)^2 + 2^{2p} (\yr)^{2p} \, \Bigl( \dfrac{h^2}{l^2} \Bigr)^{p-1} + (z)^2 = 1, \label{generalized sphere condition}
\end{equation}
and with the additional bound
\begin{equation} \label{generalized (yr,yi)-bounds}
    (\yr)^2 + (\yi)^2 = \dfrac{1}{2} \dfrac{l^2}{h^2},
\end{equation}
One can see that eqs.~(\ref{x-equation}, \ref{yr-equation}, \ref{yi-equation}, \ref{z-equation}, \ref{h-equation}, \ref{sphere condition}) are recovered for $p=1$.
Adapting the logic of the main text to the generalized eqs.~(\ref{generalized x-equation}, \ref{generalized yr-equation}, \ref{generalized yi-equation}, \ref{generalized z-equation}, \ref{generalized h-equation}, \ref{generalized sphere condition}) is straightforward, as we show below.

\begin{lemma} \label{lemma: generalized abx upper bound}
    One has
    \begin{equation} \label{eq.: generalized abx upper bound}
        \ab x(t) \ab \leq \ab x_\In \ab + \dfrac{p}{2^{p - \frac{3}{2}}} \, c l \, \Bigl( \dfrac{h_\In}{l} \Bigr)^{\! 1 - \frac{1}{p}} (t-t_\In).
    \end{equation}
\end{lemma}

\begin{proof}
    In view of eq.~(\ref{generalized x-equation}), we may write
    \begin{align*}
        \dfrac{\de}{\de t} \, \ab x \ab = \biggl[ - \ab x \ab (1 - \xi) + \mathrm{sgn}(x) \, 2^{p+1} p c \, (\yr)^{2p-1} \yi \, \Bigl( \dfrac{h^2}{l^2} \Bigr)^{p-1} \biggr] \, h.
    \end{align*}
    Then, because $\xi \leq 1$, we may write $\smash{- \ab x \ab (1- \xi) \leq 0}$.
    In view of eq.~(\ref{generalized sphere condition}), we may observe the bound
    \begin{align*}
        \ab \yr \ab^{2p - 1} \Bigl( \dfrac{h^2}{l^2} \Bigr)^{p-1} \leq \dfrac{1}{2^{2p-1}} \Bigl( \dfrac{h}{l} \Bigr)^{\! 1 - \frac{1}{p}}.
    \end{align*}
    Furthermore, in view of eq.~(\ref{generalized (yr,yi)-bounds}), we can write $\smash{\ab \yi \ab f \leq l/\sqrt{2}}$.
    Hence, we may write the differential inequality
    \begin{align*}
        \dfrac{\de}{\de t} \, \ab x \ab \leq \dfrac{p}{2^{p - \frac{3}{2}}} \, c l \, \Bigl( \dfrac{h}{l} \Bigr)^{\! 1 - \frac{1}{p}} \leq \dfrac{p}{2^{p - \frac{3}{2}}} \, c l \, \Bigl( \dfrac{h_\In}{l} \Bigr)^{\! 1 - \frac{1}{p}},
    \end{align*}
    where we also took advantage of the fact that $f$ is non-increasing.
    Then, the statement in eq.~(\ref{eq.: generalized abx upper bound}) follows immediately.
\end{proof}

\begin{lemma} \label{lemma: generalized time-interval lower bound}
    The duration of the time interval $[t_\In, t_\fin]$ is bounded from below as
    \begin{equation} \label{eq.: generalized time-interval lower bound}
        \Delta t_{\In \fin} \geq \dfrac{2^{p-\frac{5}{2}}}{p \sqrt{\xi_\fin}} \dfrac{1}{cl} \Bigl( \dfrac{l}{h_\In} \Bigr)^{\! 1 - \frac{1}{p}} \, (\xi_\fin - \xi_\In).
    \end{equation}
\end{lemma}

\begin{proof}
    Let us split the proof into two main steps.
    \begin{enumerate}[label=\roman*.]
        \item In view of eq.~(\ref{generalized z-equation}), one may write $\smash{\dot{z} = [-(1+w)/2 + \xi] \, z h \leq [-(1+w)/2 + \xi_\fin] \, z h \leq 0}$.
        Hence, because the function $z$ is non-increasing, and, because 
        \begin{align*}
            [x(t)]^2 - (x_\In)^2 = \xi(t) - \xi_\In - [(1+w)/2] \, \bigl[ [z(t)]^2 - (z_\In)^2 \bigr],
        \end{align*}
        one may write
        \begin{align*}
            [x(t)]^2 - (x_\In)^2 \geq \xi(t) - \xi_\In.
        \end{align*}
        \item In view of eq.~(\ref{eq.: generalized abx upper bound}), and because $\smash{(x)^2 \leq \xi_\fin - [(1+w)/2] \, (z)^2 \leq \xi_\fin}$, we can write
        \begin{align*}
            [x(t)]^2 - (x_\In)^2 \leq [\ab x(t) \ab + \ab x_\In \ab] \, \dfrac{p}{2^{p - \frac{3}{2}}} \, c l \, \Bigl( \dfrac{h_\In}{l} \Bigr)^{\! 1 - \frac{1}{p}} (t-t_\In) \leq \dfrac{p \sqrt{\xi_\fin}}{2^{p - \frac{5}{2}}} \, c l \, \Bigl( \dfrac{h_\In}{l} \Bigr)^{\! 1 - \frac{1}{p}} (t-t_\In).
        \end{align*}
    \end{enumerate}
    By combining the two bounds above, it is then immediate to conclude with the bound in eq.~(\ref{eq.: generalized time-interval lower bound}).
\end{proof}

One may systematically proceed with a more careful treatment of the differential inequalities, keeping into account all the information on the function $f$ in treating with the bounds on $\ab x \ab$ and $z$, as done for the case $p=1$ in app. \ref{app.: refined bounds}.

\addcontentsline{toc}{section}{References}

\bibliographystyle{JHEP}
\bibliography{refs}

\end{document}